\documentclass[journal,10pt,letterpaper]{IEEEtran}

\usepackage{epstopdf}
\usepackage{graphicx}
\usepackage{stmaryrd}
\usepackage{amsmath,amsthm,amssymb}
\usepackage{multirow,url}
\usepackage{color,cite}
\usepackage{algorithm}
\usepackage{algorithmicx}
\usepackage{algpseudocode}
\usepackage{setspace}
\usepackage{rotating}
\usepackage{bm}

\newtheorem{lem}{Lemma}
\newtheorem{thm}{Theorem}
\newtheorem{defn}{Definition}

\ifodd 0
\newcommand{\rev}[1]{{\color{blue}#1}} %revise of the text
\newcommand{\com}[1]{\textbf{\color{red} (COMMENT: #1) }} %comment of the text
\newcommand{\comg}[1]{\textbf{\color{green} (COMMENT: #1)}}
\newcommand{\response}[1]{\textbf{\color{green} (RESPONSE: #1)}} %response to comment
\else
\newcommand{\rev}[1]{#1}

\newcommand{\com}[1]{}
\newcommand{\comg}[1]{}
\newcommand{\response}[1]{}
\fi
\makeatletter
\let\@copyrightspace\relax
\makeatother

\begin{document}

\title{Amazon in the White Space: Social Recommendation Aided Distributed Spectrum Access}

\author{Xu Chen, \emph{Member, IEEE},  Xiaowen Gong, \emph{Student Member, IEEE}, Lei Yang, \emph{Member, IEEE}, and Junshan Zhang, \emph{Fellow, IEEE}
}

\maketitle
\pagestyle{empty}
\thispagestyle{empty}

\begin{abstract}
Distributed spectrum access (DSA) is challenging since an individual secondary user often has  limited sensing capabilities only. One key insight is that channel recommendation among secondary users can help to take advantage of the inherent correlation structure of spectrum availability in both time and space, and enable users to obtain more informed spectrum opportunities. With this insight, we advocate to leverage the wisdom of crowds, and devise  social recommendation aided DSA mechanisms to orient secondary users to make more intelligent spectrum access decisions, for both strong and weak network information cases. We start with the strong network information case where secondary users have the statistical information.  To mitigate the difficulty due to the curse of dimensionality in  the stochastic game approach, we take the one-step Nash approach and cast the social recommendation aided DSA decision making problem at each time slot as a strategic game. We show that it is a potential game, and then  devise an algorithm to achieve the Nash equilibrium by exploiting its finite improvement property. For the weak information case where secondary users do not have the statistical information, we develop a distributed reinforcement learning mechanism for social recommendation aided DSA based on the local observations of secondary users only. Appealing to the maximum-norm contraction mapping, we also derive the conditions under which the distributed mechanism converges and characterize the equilibrium therein. Numerical results  reveal that the proposed social recommendation aided DSA mechanisms can achieve superior performance using real social data traces and its performance loss in the weak network information case is insignificant, compared with the strong network information case.
\end{abstract}


\allowdisplaybreaks

\section{Introduction}\label{introduction}
Distributed spectrum access (DSA) is a promising
technique to alleviate the problem of spectrum under-utilization. It is envisaged that future generation wireless devices would be able to sense the
communication environment, through built-in spectrum sniffers as well as other   devices that are
wirelessly-connected in the vicinity, and cognitive radio users (secondary
users) can leverage DSA to opportunistically access the licensed
channels owned by legacy spectrum holders (primary
users).

\begin{figure}
\centering
\includegraphics[scale=1.0]{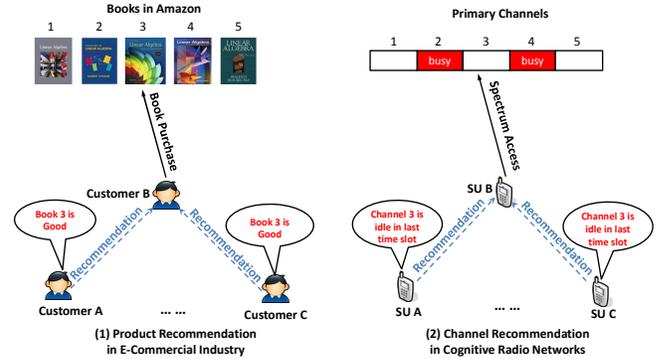}
\caption{\label{fig:rec}A comparison of recommendation systems in e-commercial industry and cognitive radio networks. Similar to the product recommendation, for the channel recommendation we can treat secondary users (SUs) as customers and the channels as products. For example, SUs A and C who have just successfully accessed channel $3$ recommend this channel to their common friend, SU B. Accordingly, SU B can achieve more informed decision making for spectrum access. We will outline the major differences in Section \ref{introduction}. }
\end{figure}

Efficient spectrum access decision making among secondary users is challenging for
several reasons: (1) spectrum opportunities for secondary users often change fast in frequency, time, and space, due to stochastic traffic and random activity patterns of primary users; (2) individual secondary users often have limited observations of the entire network environments due to the limited
spectrum sensing capability.

In this paper, we explore social intelligence of secondary users for DSA decision makings, through cooperative social interactions. One main motivation for considering social intelligence is to overcome the challenges due to incomplete network
information and limited capability of individual secondary users. Intuitively, usage-based channel recommendation offers efficient proactive sensing, by leveraging the wisdom of the crowd in the vicinity. Specifically, we will study social recommendation aided DSA mechanisms, where a secondary user provides
the recommendation (review) of the channel it just accessed to nearby secondary users who have social trust with it (e.g., friends). As a result, a secondary user can collect the channel recommendations from nearby social friends. When a secondary user needs to transmit data, it can check the recommendations from
other secondary users with social trust and choose a channel to access accordingly\footnote{Clearly, the performance of social recommendation based DSA hinges
heavily on the number of trusted users in the vicinity, and we will elaborate
further on this.}.

This study is inspired by the recommendation
system in the e-commerce industry  but goes beyond it by addressing unique challenges arised in DSA.  For example, in Amazon.com existing owners of various products can
provide recommendations (reviews), so
that other potential customers can pick the products
that best suit their needs. For cognitive radio networks, spectrum usage exhibits correlation in both the temporal and spatial domains that can be exploited for designing efficient
channel recommendation system: (1) \emph{temporal correlation}: spectrum availabilities are correlated in time, indicating that
very recent observations can be useful in the near future \cite{wellens2009empirical}; (2) \emph{spatial correlation}: secondary users
close to one another may experience similar spectrum availabilities \cite{wellens2009spatial}.
As illustrated in Figure \ref{fig:rec}, we can treat secondary users as ``customers" and the
channels as ``goods", and use recommendations to guide secondary users to find the channels they need. Such a channel recommendation system enables secondary users to take advantages of the correlations in both time and space and hence make more informed decisions.

The underlying rationale for channel recommendation is built upon the insight that the hand-held devices
are typically carried by human beings and the intrinsic social relationships among users can be exploited to promote effective and trustworthy cooperation. Indeed, with the explosive
growth of online social networks such as Facebook and
Twitter, more and more people are actively involved in
online social interactions, and social connections among
people are being extensively broadened. This has opened
up a new avenue to integrate the social interactions for
cooperative networking design.

We caution that different from  recommendation in e-commercial systems, to embrace the pronounced benefits of social recommendation for dynamic spectrum access, one needs to take into account the following unique characteristics of cognitive radio networks: (1) \emph{time-varying channel occupancies}: since primary users' activities are dynamic, the accessibility of each channel may change with time, and this is in sharp contrast to the product quality in e-commercial systems; (2) \emph{user-specific  channel recommendations}: different from many e-commercial systems such as Amazon in which the recommendations (reviews) are  accessible online by any customer, in cognitive radio networks, due to the diverse social relationships among secondary users, different users may receive different channel recommendations in a time slot; (3) \emph{interference over recommended channels}: in e-commercial systems, different customers
can purchase the same commodity by adopting the same recommendation, without any confliction. In cognitive radio networks, however, a recommended channel may be accessed by multiple secondary users, thus causing interference among them. We hence need to balance between exploitation of recommended channels and exploration of un-recommended channels.

We tackle the above technical challenges by considering two different scenarios where secondary users may have strong and weak information, respectively. For the strong information case where secondary users have the statistical information, we model social recommendation aided DSA problem at each time slot as a strategic game and derive the Nash equilibrium solution. For the weak information case where secondary users do not have the statistical information, we develop a distributed reinforcement learning scheme such that secondary users adjust their spectrum access policy based on their local observations only.

Intuitively, when a secondary user has a larger number of trusted users in the proximity for channel recommendation (e.g., a user would often stay in some regular locations such as work place and home), then the user can achieve better performance. However, when the channel recommendation information is limited, each secondary user will make the DSA decision based on either the prior channel statistical information (for the strong network
information case) or its learned spectrum access experience in the unrecommended states (for
the weak network information case). This ensures that the proposed social recommendation aided DSA mechanisms can also achieve good performance, even when the number of trusted users in proximity is relatively small.

\subsection{Summary of Main Contributions}

The main contributions of this paper are as follows:
\begin{itemize}
\item \emph{Social Recommendation Aided DSA Framework}: We propose a social recommendation aided DSA framework, which exploits the temporal and spatial correlations for spectrum utilization and leverages social trust among secondary users for cooperative channel recommendation, to overcome the challenges due to limited observations of network environment by individual secondary users. This framework highlights the interplay between the physical coupling among secondary users for DSA and the exploitation of social tie structure among users to stimulate effective and trustworthy cooperation for channel recommendation.

\item \emph{DSA with Strong Network Information}: We study the social recommendation aided DSA mechanism design with strong network information. To mitigate the difficulty due to the curse of dimensionality in  the stochastic game  approach \cite{bernstein2002complexity}, we take the one-step Nash approach (i.e., a myopic approach) \cite{goode2010pursuit} and model the social recommendation aided DSA decision making problem at each time slot as a strategic game. We show that the game is a potential game by constructing a potential function based on the physical and social coupling structure of social recommendation aided DSA. By exploiting the finite improvement property of the potential game, we further devise an algorithm that can achieve the Nash equilibrium.

\item \emph{DSA with Weak Network Information}: We investigate the social recommendation aided DSA mechanism design with weak network information. We develop a distributed reinforcement learning mechanism based on the local observations of secondary users only. By resorting to the tool of maximum-norm contraction mapping, we derive the conditions under which the learning mechanism converges to the solution of the fixed point equation. We further characterize the equilibrium and show that the learning mechanism achieves the maximum expected throughput approximately, given users' observed channel recommendation states.
\end{itemize}

\subsection{Related Work}\label{relatedwork}
The social aspect is now becoming a new and important dimension for
communication system design \cite{kayastha2011applications}. As the development
of online and mobile social networks such as Facebook and Twitter, more and more
real-world data and traces of human social interactions are being generated. This enables researchers
and engineers to observe, analyze, and incorporate the social
factors into engineering system design in a way never previously
possible \cite{kayastha2011applications}. Gao \emph{et al.} in \cite{gao2009multicasting} exploited social structures such as social
community to design efficient data forwarding and routing algorithms in delay tolerant networks. Hui \emph{et al.} in \cite{hui2011bubble} used the social betweenness and centrality as the forwarding metric. Costa \emph{et al.} in \cite{costa2008socially} proposed predictions based on metrics of social interaction to identify the best information carriers for content publish-subscribe. Han \emph{et al.} in \cite{han2012mobile} utilized the social influence phenomenon to devise efficient data dissemination mechanisms for mobile networks. Two key social phenomena of social trust and social reciprocity are exploited for cooperative D2D communication mechanism design in \cite{chen2013social}. In this paper, we explore social intelligence of secondary users for DSA via channel recommendation.

Recent works \cite{li2010customer,ACR2013} study the channel recommendation schemes, by assuming that the network information is complete and all the secondary users are cooperative, i.e., the interest of all users is aligned and  the underlying social graph for channel recommendation is complete (i.e., all the users are willing to help each other). In this paper, we eliminate these restrictive assumptions and consider that secondary users have strong/weak but incomplete network information. More importantly, we leverage the endogenous incentive which comes from the intrinsic social relationships among secondary users to promote effective and trustworthy cooperation for channel recommendation. We hence consider a general social graph among secondary users for channel recommendation. For the extreme case of the complete social graph \cite{li2010customer,ACR2013}, since a secondary user can receive the channel recommendations from any other users, all users would have the same observations of the channel recommendation states in each time slot. While for the case of a  general social graph, due to diverse social relationships among secondary users for channel recommendations,  users would have heterogenous observations of the channel recommendation states in each time slot, making the spectrum access mechanism design much more challenging.

%The rest of the paper is organized as follows. We first introduce the system model in Section \ref{system}.  We then thoroughly investigate the social recommendation aided DSA mechanism design with strong and weak network information in Sections \ref{sec:Social-Recommendation-Based}
%and \ref{sec:Social-Recommendation-Based-1}, respectively. We evaluate the performance
%of social recommendation aided DSA by simulations in Section \ref{Sim}, and finally conclude in Section \ref{Con}.

\section{System Model}\label{system}

In this section we first introduce the system model of social recommendation
aided DSA, which can be projected onto two domains: the
physical domain and the social domain. In the physical domain, different
secondary users have different physical coupling due to their heterogeneous
interference relationships. In the social domain, different secondary
users have different channel recommendation relationships due to the
diverse social tie structure among secondary users. We next discuss both
physical and social domains in detail.
%
%\begin{figure}
%\centering
%\includegraphics[scale=0.5]{Markov}
%\caption{\label{fig:Two-states-Markovian}Two-state Markovian channel model}
%\end{figure}

\subsection{Physical Graph Model\label{sub:Physical-Graph-Model}}

We consider a dynamic spectrum access network with a set $\mathcal{M}=\{1,2,...,M\}$
of heterogeneous primary channels. A set $\mathcal{N}=\{1,2,...,N\}$ of
secondary users try to access these channels, in a distributed manner,
when the channels are not occupied by primary (licensed) transmissions.

To capture the physical coupling, we construct the interference graph
$\mathcal{G}^{p}=\{\mathcal{N},\mathcal{E}^{p}\}$ based on the interference
relationships among secondary users. Here the vertex set $\mathcal{N}$
is the secondary user set, and the edge set $\mathcal{E}^{p}=\{(j,k):d_{j,k}\leq\delta,\forall j,k\neq j\in\mathcal{N}\}$
is the set of interference edges with $d_{j,k}$ being the distance
between secondary users $j$ and $k$ and $\delta$ denoting the interference
range. If there is an interference edge between two secondary users,
then they cannot successfully transmit their data on the same idle
channel simultaneously due to collision. In the sequel, we denote
the set of interfering users with user $n$ as $\mathcal{N}_{n}^{p}=\{j:(n,j)\in\mathcal{E}^{p},\forall j\in\mathcal{N}\}$.

We consider a time-slotted system model as follows:

(1) \emph{Channel state}: For each primary channel $m$, the channel state
at a time slot $t$ is given as $S_{m}(t)=-1$ if channel $m$ is busy, and $S_{m}(t)=1$ if channel $m$ is idle.
%\[
%S_{m}(t)=\begin{cases}
%-1, & \mbox{channel $m$ is busy,}\\
%1, & \mbox{channel $m$ is idle.}
%\end{cases}
%\]

(2) \emph{Channel state transition}: Similar to many studies
(e.g., \cite{zhao2007decentralized} and the references therein), the state change of a channel $m$ is modeled as a two-state Markovian process\footnote{The proposed mechanisms can be extended to other channel state models (e.g., high-order Markovian channel model \cite{chen2010prediction}) by following the ideas developed in this paper.}.
We denote the channel state probability distribution at time $t$
as $\boldsymbol{c}_{m}(t)\triangleq(Pr\{S_{m}(t)=-1\},Pr\{S_{m}(t)=1\})$. It follows that  $\boldsymbol{c}_{m}(t)=\boldsymbol{c}_{m}(t-1)\Gamma_{m},\forall t\geq1$,
with the transition matrix $\Gamma_{m}=\left[\begin{array}{cc}
1-\lambda_{m} & \lambda_{m}\\
\mu_{m} & 1-\mu_{m}
\end{array}\right],$ where $0\le\lambda_{m}\le 1$ and $0\le\mu_{m}\le 1$ are the state transition probabilities. Furthermore, the average channel idle probability $\gamma_{m}$ of
a channel $m$ can be obtained from the stationary distribution of
the Markov chain as $\gamma_{m}=\frac{\lambda_{m}}{\lambda_{m}+\mu_{m}}.$
%\begin{eqnarray}
% &  & \gamma_{m}=\frac{\lambda_{m}}{\lambda_{m}+\mu_{m}}.\label{eq:sd-2}
%\end{eqnarray}

(3) \emph{Channel throughput}: For a secondary user $n$,
its realized channel throughput $b_{m}^{n}(t)$ on an idle channel
$m$ evolves according to an i.i.d. random
process with a finite mean $B_{m}^{n}$. Here the channel throughput
can be user-specific, which is useful to model the scenario that users
may adopt heterogeneous coding/modulation transmission schemes and
experience different channel conditions due to their local environmental
effects such as fading. For example, in a Rayleigh fading channel
environment, the channel gain in each time slot is a random variable
that follows the exponential distribution.

(4) \emph{Channel contention}: To resolve the transmission collision when
multiple interfering secondary users access the same idle channel,
we adopt the random access scheme for medium access control. Specifically,
each secondary user $n$ will contend for the idle channel with a
probability of $p_{n}$ and a collision occurs if multiple interfering
users contend for the same channel at the same time. In this case,
the probability that secondary user $n$ can successfully grab the
idle channel for data transmission is given as $q_{n}(\boldsymbol{a})=p_{n}\prod_{k\in\mathcal{N}_{n}^{m}(\boldsymbol{a})}(1-p_{k}),$ where $\boldsymbol{a}=(a_{1},...,a_{N})$ denotes the channel selection
profile of the secondary users and $\mathcal{N}_{n}^{m}(\boldsymbol{a})\triangleq\{k:a_{k}=a_{n}\mbox{ and }k\in\mathcal{N}_{n}^{p}\}$
is the set of interfering users that choose the same channel $m$
as user $n$.

\subsection{Social Graph Model}\label{socialgraph}

In order to carry out channel recommendation, we assume that
there exists a common control channel for the information exchange
among the secondary users%
\footnote{Please refer to \cite{cormio2010common} for the details on how to set up and maintain
a reliable common control channel in cognitive radio networks.%
}. When the common control channel approach is not feasible, we can
adopt the device-to-device (D2D) communication approach, such that
secondary users equipped with the radio interfaces such as blue-tooth/WiFi-direct
can communicate with each other directly for the channel recommendation. Alternatively, users can publish their recommendations in a cloud-based platform that other users have access to.

Since information exchange would incur overhead such as energy consumption, to achieve trustworthy collaboration among secondary users for channel
recommendation, we leverage the social trust among users. The underlying
rationale of using social trust is that hand-held devices are carried
by human beings and the knowledge of human social ties (e.g., kinship,
friendship, or colleague relationship) can be utilized to achieve
effective and trustworthy collaboration for channel recommendation. Moreover,  building upon the social trust among secondary users, we can prevent the potential attacks of releasing false channel recommendations by untrusted/unacquainted users and enhance the security level of DSA.

Specifically, we introduce the social graph $\mathcal{G}^{s}=\{\mathcal{N},\mathcal{E}^{s}\}$
to model social recommendation relationships due to the social trust
among secondary users. Here the vertex set is the same as the user
set $\mathcal{N}$ and the edge set is given as $\mathcal{E}^{s}=\{(n,m):e_{nm}^{s}=1,\forall n,m\in\mathcal{N}\}$
where $e_{nm}^{s}=1$ if and only if users $n$ and $m$ have social
trust between each other and can also exchange information with each
other. In practice we can adopt the private matching protocol in \cite{von2008veneta} such that two secondary users can locally identify the social relationships among them in a privacy-preserving manner. Note that the size of data packet for encapsulating recommendation information is very small. For example, for a 8-bit packet, we can use the first 7 bits to denote the channel ID (in this case we have up to 128 channel IDs available) and the last bit to indicate the idle/busy state of the channel. Thus, the recommendation information exchange can be carried out quickly during a slot.

\subsection{Social Recommendation Aided DSA}

Based on the physical and social graph models above, we devise
the social recommendation aided DSA mechanism.
A key idea is that secondary users provide the channel recommendations (reviews) of
the channels they have just accessed to their social neighbors for spectrum
access decision marking. Specifically, each secondary user $n\in\mathcal{N}$
executes the following five stages during each time
slot $t$:

(1) \emph{Spectrum sensing}: Sense the channel based on channel selection
decision made at the end of the previous time slot. The purpose of
spectrum sense is to detect the presence of primary users, in order
to avoid causing harmful interference to primary users.

(2) \emph{Channel contention}: If the channel sensing result is idle,
compete for the channel with the random access mechanism described
in Section \ref{sub:Physical-Graph-Model}.

(3) \emph{Data transmission}: Transmit data packets if user $n$ successfully
grabs the channel.

(4) \emph{Channel recommendation}: User $n$ informs its social neighbors
$m\in\mathcal{N}_{n}^{s}$ about the channel recommendation (review),
which contains the channel ID and the channel quality (idle or busy)
that it has just selected to access. At the same time, user $n$ collects
the channel recommendations from its social neighbors. Based on the
recommendations, user $n$ then has an observation of the channel
availability of all channels as $\boldsymbol{I}_{n}(t)=(I_{1}^{n}(t),...,I_{M}^{n}(t))$
at current time slot $t$, where
\begin{eqnarray*}
 &  & I_{m}^{n}(t)=\\
 &  & \begin{cases}
1, & \mbox{if \ensuremath{\{k\in\mathcal{N}_{n}^{s}}: \ensuremath{a_{k}(t)=m} and \ensuremath{S_{m}(t)=1\}\neq\emptyset}\ensuremath{,}}\\
-1, & \mbox{\mbox{if \ensuremath{\{k\in\mathcal{N}_{n}^{s}}: \ensuremath{a_{k}(t)=m} and \ensuremath{S_{m}(t)=-1\}\text{\ensuremath{\neq}}\emptyset},}}\\
0, & \mbox{if \ensuremath{\ensuremath{\{k\in\mathcal{N}_{n}^{s}}:\ensuremath{a_{k}(t)=m}\}=\emptyset}.}
\end{cases}
\end{eqnarray*}
That is, we have $I_{m}^{n}(t)=1$ if channel $m$ is recommended
as idle by some social neighbors, $I_{m}^{n}(t)=-1$ if channel $m$
is recommended as busy by some social neighbors, and $I_{m}^{n}(t)=0$
if no recommendations about channel $m$ are available. We call $I_{m}^{n}(t)$
as the channel recommendation state in the sequel.

(5) \emph{Channel selection}: Based on the channel recommendation states
$\boldsymbol{I}_{n}(t)$, user $n$ will make the decision
of choosing a channel $a_{n}(t+1)$ to access for the next time slot
$t+1$ according to the algorithms in Sections \ref{sec:Social-Recommendation-Based}
and \ref{sec:Social-Recommendation-Based-1}.

For the ease of exposition, in this paper we use the OR rule as study case to fuse recommendations from social neighbors, i.e., the recommendation state of a channel $I_{m}^{n}(t)=1$ ($I_{m}^{n}(t)=0$, respectively) as long as the channel is  recommended as idle (busy, respectively) by at lease one social neighbor. This is reasonable due to the social trust among the users such that they are willing to recommend channels cooperatively and truthfully. Our model can also apply when other rules are adopted for recommendation information fusion. For example, to further enhance the robustness, we can use the majority voting rule such that the recommendation state of a channel $I_{m}^{n}(t)=1$ ($I_{m}^{n}(t)=0$, respectively) if the number of idle (busy, respectively) recommendations by the social neighbors is major (e.g., greater than $50\%$). Furthermore, in this case we focus on designing efficient spectrum access strategy given the fused recommendation information and hence do not consider the collusion effect in information fusion (e.g., a group of users collude to broadcast biased recommendation information). How to design a collusion-resistant recommendation information fusion scheme is very challenging and will be addressed in a future work.

We shall emphasize that the algorithm design in Sections \ref{sec:Social-Recommendation-Based} and \ref{sec:Social-Recommendation-Based-1} does not simply let the users  choose the recommended channels blindly. A key step here is
how to achieve efficient distributed decision makings among the secondary users based on the channel recommendation states $(\boldsymbol{I}_{n}(t),\forall n\in\mathcal{N})$. This is challenging because: (1) Although spectrum availabilities are correlated in time, the channel states might change later due to the dynamic activities of primary users;  (2) due
to the diverse social relationships among secondary users, different users may have different number of social friends nearby and hence have different channel recommendation states $\boldsymbol{I}_{n}(t)$. Therefore, secondary users may have heterogeneous views of the network environment; (3)
subject to the physical coupling among secondary users, if multiple
interfering users choose to access the same recommended idle channel, they may
cause severe collision among each other, leading to a poor system performance. We hence need to balance between exploitation of recommended channels and exploration of un-recommended channels.
In the following Sections  \ref{sec:Social-Recommendation-Based} and \ref{sec:Social-Recommendation-Based-1}, we will address these challenges by
considering the cases that the secondary users have strong and weak
network information, respectively. In Section \ref{Discusssion}, we will discuss the security and privacy issues of social recommendation aided DSA. Similar to many studies in dynamic spectrum access (e.g., \cite{huang2011decentralized,zhang2013data,chang2007optimal}), to enable tractable analysis and get useful insights, in the paper we consider the static user case in which the physical graph $\mathcal{G}^{p}$ and the social graph $\mathcal{G}^{s}$ are fixed. This is also motivated by that users' locations would change at the time scale of seconds/minutes and the proposed algorithms are implemented at a much smaller time scale of microseconds. The dynamic case that both users' locations and the physical/social graph change very fast is very challenging and will be considered in a future work.

\section{Social Recommendation Aided DSA: The Strong Information Case}\label{sec:Social-Recommendation-Based}

In this section, we first consider the social recommendation aided
DSA mechanism design with strong network
information. For this case, a secondary user $n$ has strong knowledge
of the network environment, including the channel state transition
matrix $(\Gamma_{m},\forall m\in\mathcal{M})$, the mean channel throughput
$(B_{m}^{n},\forall m\in\mathcal{M})$, and the channel contention
probabilities $(p_{k},\forall k\in\mathcal{N}_{n}^{p})$ of its interfering
users. Note that in the strong network information case, a secondary user
does not have the complete network information, since a user
may not know other secondary users' mean channel throughputs and channel recommendation states. The motivation of considering the strong network information case is two-fold: (1) the strong network information case  can
be relevant to the scenario that there exists a supporting infrastructure
(e.g., secondary access point) to collect and analyze the statistical
network information; (2) the strong information case can serve as the benchmark for the weak information case in Section \ref{sec:Social-Recommendation-Based-1}.

At a time slot $t$, according to the channel state transition model in Section \ref{sub:Physical-Graph-Model},
a secondary user $n$ can compute the
channel idle probability $\omega_{m}^{n}(t+1)$ for the next time slot $t+1$, based on a channel $m$'s recommendation state $I_{m}^{n}(t)$ as
\[\omega_{m}^{n}(t+1)=\theta_{m}^{n}(I_{m}^{n}(t)),\]
where $\theta_{m}^{n}(\cdot)$ is the mapping function defined as follows:
\[
\theta_{m}^{n}(I_{m}^{n}(t))\triangleq\begin{cases}
1-\mu_{m}, & \mbox{if \ensuremath{I_{m}^{n}(t)=1,}}\\
\lambda_{m}, & \mbox{if \ensuremath{I_{m}^{n}(t)=-1,}}\\
\gamma_{m}, & \mbox{if \ensuremath{I_{m}^{n}(t)=0.}}
\end{cases}
\]
Based on the channel idle probabilities $(\omega_{m}^{n}(t),\forall m\in\mathcal{M})$ of all channels,
a secondary user $n$ will then make spectrum access decision for the next
time slot $t+1$. Let $a_{n}(t+1)$ denote the channel selection decision
of a secondary user $n$ and $\boldsymbol{a}(t+1)=(a_{1}(t+1),...,a_{N}(t+1))$
be the channel selection decisions of all secondary users. We
can then obtain the expected throughput of secondary user $n$ at
the next time slot as
\begin{align}
    U_{n}(\boldsymbol{a}(t+1))= & \omega_{a_{n}(t+1)}^{n}(t+1)B_{a_{n}(t+1)}^{n}\nonumber \\
  & \times p_{n}\prod_{k\in\mathcal{N}_{n}^{a_{n}(t+1)}(\boldsymbol{a}(t+1))}(1-p_{k}).\label{eq:U1}
\end{align}
Here $\mathcal{N}_{n}^{a_{n}(t+1)}(\boldsymbol{a}(t+1))\triangleq\{k:a_{k}(t+1)=a_{n}(t+1)\mbox{ and }k\in\mathcal{N}_{n}^{p}\}$
is the set of interfering secondary users that select the same channel
as secondary user $n$. Note that for the ease of exposition, similar to many previous studies in spectrum sensing (e.g., \cite{maskery2009decentralized,ahmad2009optimality,hossain2009dynamic} and the references therein), we assume that secondary users would experience the same spectrum availability on the same channel during a time slot. This can correspond to the practical case that the primary users (e.g., TV broadcast towers) have a much larger wireless transmission coverage area than the secondary users. Nevertheless, the proposed model in this paper can be extended for the case that secondary users may experience the heterogenous spectrum availability, by modifying the channel idle probability computing function $\omega_{m}^{n}(t+1)$ that maps a user's received channel recommendation states to the idle probability of a channel.  For example, when a user's received recommendation states diversify, we can adopt the approach in \cite{xue2014cooperative} that exploits the correlations among users' spectrum observations to compute the channel idle probability.

\subsection{Problem Formulation}

We next consider the distributed decision making problem among the
secondary users for maximizing their throughputs. One possible approach is to formulate the DSA decision making problem as a stochastic game, such that the game state at each time slot $t$ is defined as the channel recommendation states $\{\boldsymbol{I}_{n}(t),n\in\mathcal{N}\}$ of all the users. However, due to the curse of dimensionality \cite{bernstein2002complexity} and the strong coupling among the users in both social and physical dimensions, such a stochastic game is computationally intractable and extremely difficult to analyze. To enable tractable analysis, in this paper we take the one-step Nash approach \cite{goode2010pursuit} (i.e., myopic policy), such that in each time slot $t$ we compute the Nash equilibrium for the stage game based on the current game state only. For the problem under consideration, the one-step Nash solution can be an efficient approximate solution for the stochastic game approach, since secondary users would like to diversify their channel selections at the Nash equilibrium (to mitigate congestion effect in DSA) and this would also lead to efficient channel recommendations (i.e., due to the gain of getting more recommendation information by exploring more channels) for improving the future DSA decision makings. Numerical results in Section \ref{Sim} corroborate that the one-step Nash approach can achieve superior performance.

Specifically, at a time slot $t$, let $a_{-n}(t+1)$ be the set of channel selections chosen by all
other users except user $n$ for the next time slot. Given the other users' spectrum access
decisions $a_{-n}(t+1)$, user $n$ would like to choose a channel $a_{n}\in\mathcal{M}$
to maximize its expected throughput at the next time slot, i.e.,
\[
\max_{a_{n}(t+1)\in\mathcal{M}}U_{n}\Bigl(a_{n}(t+1),a_{-n}(t+1)\Bigr),\forall n\in\mathcal{N}.
\]
The distributed nature of the problem above naturally leads to a game theoretic formulation
such that each secondary user aims to maximize its expected throughput for the next time slot.
We thus formulate the DSA decision making problem among the secondary users at time slot $t$ as a
strategic game $\Omega=(\mathcal{N},\mathcal{M},\{U_{n}\}_{n\in\mathcal{N}})$,
where the set of secondary users $\mathcal{N}$ is the set of players,
the channel set $\mathcal{M}$ is the set of strategies for each  user,
and the throughput function $U_{n}$ of each user $n$ is the payoff
function of player $n$. In the sequel, we call the game $\Omega$
as the social recommendation aided DSA game at the time slot $t$. Note that at different time slot $t$, the channel recommendation states $\{\boldsymbol{I}_{n}(t),n\in\mathcal{N}\}$ can be different and hence the underlying specific social recommendation aided DSA game can be different. Our analysis in the following section holds for a general social recommendation aided DSA game.

\subsection{Structural Property}

We next explore the property of the social recommendation aided DSA game at a time slot $t$. We first study the existence of Nash equilibrium, which
is defined as follows.
\begin{defn}
A spectrum access profile $\boldsymbol{a}^{*}=(a_{1}^{*},...,a_{N}^{*})$
is a Nash equilibrium of the social recommendation aided DSA game $\Omega$ if no secondary user can improve its expected throughput
by unilaterally changing its channel selection, i.e., $a_{n}^{*}=\max_{a_{n}\in\mathcal{M}}U_{n}(a_{n},a_{-n}^{*}),\forall n\in\mathcal{N}.$
\end{defn}

To show the the existence of Nash equilibrium, we resort to a useful
tool of potential game.
\begin{defn}
A game is called a potential game if it admits a potential function
$\Phi(\boldsymbol{a})$ such that for every $n\in\mathcal{N}$ and
$a_{-n}\in\mathcal{M}^{N-1}$, given any $a_{n},a_{n}^{'}\in\mathcal{M}$, we have
\begin{eqnarray}
 &  & \mbox{sgn}\left(U_{n}(a_{n}^{'},a_{-n})-U_{n}(a_{n},a_{-n})\right)\nonumber \\
 & = & \mbox{sgn}\left(\Phi(a_{n}^{'},a_{-n})-\Phi(a_{n},a_{-n})\right),\label{eq:potential}
\end{eqnarray}
where $\mbox{sgn}(\cdot)$ is the sign function.
\end{defn}

\begin{defn}
The event where a player $n$ changes to an action $a_{n}^{'}$ from
the action $a_{n}$ is a better response update if and only if $U_{n}(a_{n}^{'},a_{-n})>U_{n}(a_{n},a_{-n})$.
\end{defn}

\rev{An appealing property of the potential game is that it admits the \emph{finite improvement property}, such that any asynchronous better response
update (no more than one player updates the strategy at any given time) must be finite and leads to a Nash equilibrium \cite{PG1996}. This is because that due to (\ref{eq:potential}), when a user unilaterally carries out a better response update to improve its utility function, implicitly this will also lead to an increase in the potential function. Since the potential function is upper-bounded (i.e., can not keep increasing infinitely) and the strategy space is finite, the asynchronous better response update process must stop within some finite steps and reaches a Nash equilibrium wherein no further improvement by any user can be carried out.}

%To show that the social recommendation aided DSA game
%$\Omega$ is a potential game, we consider a closely related game
%$\widetilde{\Omega}=(\mathcal{N},\mathcal{M},\{\widetilde{U}_{n}\}_{n\in\mathcal{N}})$
%by the monotone transformation, where
%\begin{equation}
%\widetilde{U}_{n}(\boldsymbol{a}(t+1))=\ln U_{n}(\boldsymbol{a}(t+1)).\label{eq:U2}
%\end{equation}
%Due to property of monotone transformation such that \begin{align*}
% & \mbox{sgn}\left(U_{n}(a_{n}^{'},a_{-n})-U_{n}(a_{n},a_{-n})\right)\\
%= & \mbox{sgn}\left(\widetilde{U}_{n}(a_{n}^{'},a_{-n})-\widetilde{U}_{n}(a_{n},a_{-n})\right),\end{align*}
%we have the following result.
%\begin{lem}
%\label{lem:If-the-modified1}If the modified game $\widetilde{\Omega}$
%is a potential game, then the original social recommendation aided
%DSA game $\Omega$ is also a potential game with the same
%potential function.
%\end{lem}
%
%For the modified game $\widetilde{\Omega}$, we show that it is a
%potential game with the following potential function
%\begin{align}
%& \Phi(\boldsymbol{a}(t+1))=  \sum_{n=1}^{N}-\ln(1-p_{n}) \nonumber \\
%& \times\left(\frac{\sum_{i\in\mathcal{N}_{n}^{p}(\boldsymbol{a}(t+1))}\ln(1-p_{i})}{2}+\ln\left(\omega_{a_{n}(t+1)}^{n}(t+1)B_{a_{n}}^{n}p_{n}\right)\right).\label{eq:PPPP}
%\end{align}
%
%\begin{lem}
%\label{lem:The-modified-Game}The modified game $\widetilde{\Omega}$
%is a potential game with the potential function $\Phi(\boldsymbol{a}(t+1))$
%in (\ref{eq:PPPP}).\end{lem}
We can show that the social recommendation aided DSA
game $\Omega$ is indeed a potential game with a potential function given as
\begin{align}
& \Phi(\boldsymbol{a}(t+1))=  \sum_{n=1}^{N}-\ln(1-p_{n})\times \nonumber \\
& \left(\frac{\sum_{i\in\mathcal{N}_{n}^{p}(\boldsymbol{a}(t+1))}\ln(1-p_{i})}{2}+\ln\left(\omega_{a_{n}(t+1)}^{n}(t+1)B_{a_{n}(t+1)}^{n}p_{n}\right)\right).\label{eq:PPPP}
\end{align}
\begin{thm}
\label{thm:The-social-recommendation}The social recommendation aided
DSA game $\Omega$ is a potential game with the potential
function $\Phi(\boldsymbol{a}(t+1))$ in (\ref{eq:PPPP}), and hence possesses the finite improvement property and always
has a Nash equilibrium.
\end{thm}

The proof is relegated to the appendix. \rev{Intuitively, we can interpret the potential function above from a graphical point of view. Specifically, given a strategy profile of all secondary users $\boldsymbol{a}(t+1)$, we can construct a ``virtual" interaction graph such that each user $n$ is a vertex in the graph with the vertex weight of $-\ln(1-p_n)\ln\left(\omega_{a_{n}(t+1)}^{n}(t+1)B_{a_{n}(t+1)}^{n}p_{n}\right)$, and there exists an undirected edge between two interfering users $n$ and $m$ that access the same spectrum (i.e., $a_n(t+1)=a_m(t+1)$) and the edge weight is $-\ln(1-p_n)\ln(1-p_m)$. Thus, the potential function in (\ref{eq:PPPP}) represents the total weight of the interaction graph such that $\sum_{n=1}^{N}-\ln(1-p_{n})\ln\left(\omega_{a_{n}(t+1)}^{n}(t+1)B_{a_{n}(t+1)}^{n}p_{n}\right)$ is the total weight of all the vertices and $\frac{\sum_{n=1}^{N}\sum_{i\in\mathcal{N}_{n}^{p}(\boldsymbol{a}(t+1))}-\ln(1-p_{n})\ln(1-p_{i})}{2}$ is the total weight of all the edges. When a user $n$ unilaterally changes its channel from $a_{n}(t+1)$ to $a_{n}^{'}(t+1)$, the change of the total weight of the interaction graph consists of two parts: the first part is due to the change in the vertex weight of user $n$, i.e., \begin{align*}
-\ln(1-p_n)&  \left( \ln\left(\omega_{a_{n}(t+1)}^{n}(t+1)B_{a_{n}(t+1)}^{n}p_{n}\right)\right. \\
          &  \left.  -\ln\left(\omega_{a_{n}^{'}(t+1)}^{n}(t+1)B_{a_{n}^{'}(t+1)}^{n}p_{n}\right)\right),
\end{align*}
and the second part is due to the change in the edge weight by adding new edges between user $n$ and the new interfering users $\mathcal{N}_{n}^{p}(\boldsymbol{a}^{'}(t+1))$ and removing old edges between user $n$ and the previous interfering users $\mathcal{N}_{n}^{p}(\boldsymbol{a}(t+1))$, i.e.,\begin{align*}
-\ln(1-p_n)&  \left( \sum_{i\in\mathcal{N}_{n}^{p}(\boldsymbol{a}^{'}(t+1))}\ln(1-p_{i})-\sum_{i\in\mathcal{N}_{n}^{p}(\boldsymbol{a}(t+1))}\ln(1-p_{i})\right).
%-\ln(1-p_n)&  \left( \sum_{i\in\mathcal{N}_{n}^{p}(\boldsymbol{a}^{'}(t+1))}\ln(1-p_{i})\right. \\
%          &  \left.  -\sum_{i\in\mathcal{N}_{n}^{p}(\boldsymbol{a}(t+1))}\ln(1-p_{i})\right).
\end{align*}
Thus, the change of the total weight of the interaction equals to $-\ln(1-p_n) \left(\ln (U_{n}(a_{n}^{'},a_{-n}))-\ln (U_{n}(a_{n},a_{-n}))\right)$, and hence the condition in (\ref{eq:potential}) is satisfied. That is, when a user $n$ unilaterally changes its channel from $a_{n}(t+1)$ to $a_{n}^{'}(t+1)$ to improve its utility function, this will also lead to an increase in the total weight of the interaction graph (i.e., the potential function). }

%Intuitively, we can regard the term $\ln\left(\omega_{a_{n}(t+1)}^{n}(t+1)B_{a_{n}}^{n}p_{n}\right)$ in the potential function (\ref{eq:PPPP}) as user $n$'s utility gain (in terms of log utility) when using the channel $a_{n}$ by itself alone, and the term $\frac{\sum_{i\in\mathcal{N}_{n}^{p}(\boldsymbol{a}(t+1))}\ln(1-p_{i})}{2}$ reflects user $n$'s utility degradation due to the presence of interfering users in its proximity who also access the same channel. Thus, $\ln\left(\omega_{a_{n}(t+1)}^{n}(t+1)B_{a_{n}}^{n}p_{n}\right)+\frac{\sum_{i\in\mathcal{N}_{n}^{p}(\boldsymbol{a}(t+1))}\ln(1-p_{i})}{2}$ describes user $n$'s net utility, and hence the potential function (\ref{eq:PPPP}) characterizes the total weighted net utility by all users in the system.

The result in Theorem \ref{thm:The-social-recommendation} implies
that any asynchronous better response update process is guaranteed
to reach a Nash equilibrium within a finite number of iterations.
This motivates the algorithm design for computing the Nash equilibrium for each time slot $t$
in the following section.

\subsection{Computing Nash Equilibrium}

We next propose a scheme to compute a Nash equilibrium in Algorithm
\ref{alg:Online-Distributed-Channel}, which is carried out at
the channel selection stage of each time slot $t$.
The key idea of algorithm design is to let secondary users asynchronously
improve their channel selections according to the finite improvement
property of the social recommendation aided DSA game
$\Omega$.

\begin{algorithm}[tt]
\begin{algorithmic}[1]
\State \textbf{initialization:}
\State \hspace{0.4cm} \textbf{set} initial channel selection profile $\boldsymbol{a}^{0}(t+1)=\boldsymbol{a}(t)$.
\State \hspace{0.4cm} \textbf{set} the iteration index $l=0$.
\State \textbf{end initialization\newline}

\While{$\boldsymbol{a}^{l}(t+1)$ is not a Nash equilibrium}
        \State  Secondary user $n=1+$$(l$ mod $N$) \textbf{choose} the best channel $a_{n}^{l+1}(t+1)$ that maximizes the expected throughput according to (\ref{eq:bs1}).
%            \State \textbf{broadcast} the channel selection update $a_{n}^{l+1}(t+1)$ to other secondary users.
        \State \textbf{set} the index $l=l+1$.
\EndWhile

\end{algorithmic}
\caption{\label{alg:Online-Distributed-Channel} Nash equilibrium computation algorithm for social recommendation aided DSA
with strong network information}
\end{algorithm}

Further, we assume that when a secondary user enters the system,
the supporting infrastructure (e.g., secondary access point) will
assign it with a unique ID, say index as $1,2,3,...$\footnote{When the supporting infrastructure is not available, we can achieve asynchronous the best response update by allowing that secondary users adopt the random backoff mechanism such as CSMA to content for the channel selection update opportunities. This can also converge to a Nash equilibrium, due to the finite improvement property of the social recommendation aided DSA game.}. For the initial iteration $l=0$,
we set the initial channel selection decision profile for the next time slot $t+1$ as the channel
selections adopted by the secondary users at the current time slot $t$, i.e., $\boldsymbol{a}^{l}(t+1)=\boldsymbol{a}(t)$, for $l=0$.
Then for each iteration $l=0,1,2,...$, a secondary user $n=1+$$(l$ mod $N$) is
scheduled by the infrastructure to carry out the best response update
in turn (according to the assigned IDs). Specifically, the scheduled
secondary user $n$ will select a channel $a_{n}^{l}(t+1)$ that
maximizes its expected throughput for the next time slot $t+1$ as
\begin{align}
a_{n}^{l}(t+1) =  \arg\max_{a\in\mathcal{M}}  &  U_{n} \Bigl( a,a^{l}_{-n}(t+1)\Bigr),\label{eq:bs1}
\end{align}
where $a^{l}_{-n}(t+1)$ denote the current channel selection decisions of all the other users except user $n$ at the current iteration $l$. According to (\ref{eq:U1}), to carry
out the best response update in (\ref{eq:bs1}),  a secondary user $n$ needs to observe the current channel selection decisions
of its interfering users $k\in\mathcal{N}_{n}^{p}$, besides the knowledge of its own mean channel throughput, the channel state transition probabilities, and its own channel recommendation states. This can be achieved
such that each secondary user $k$ broadcasts the channel selection $a_{k}^{l+1}(t+1)$ when its decision has been updated. Such a best
response update procedure continues until that a Nash equilibrium
is achieved, i.e., no secondary users can carry out any channel selection
update (i.e., $\boldsymbol{a}^{l+1}(t+1)=\boldsymbol{a}^{l}(t+1)$).
According to the finite improvement property of the social recommendation
based DSA game $\Omega$, such an asynchronous best response
update process is guaranteed to achieve a Nash equilibrium within a finite number of iterations. Note that during each iteration a user will choose the best channel among $M$ channels, and hence the computational complexity of Algorithm \ref{alg:Online-Distributed-Channel} at one iteration is $\mathcal{O}(M)$. Now suppose it takes $L$ iterations for the algorithm to converge. Then the total computational complexity of Algorithm \ref{alg:Online-Distributed-Channel} is $\mathcal{O}(LM)$. Numerical results show that the algorithm can converge quickly, with the number of iterations $L<2N$. In this case, the computational complexity of Algorithm \ref{alg:Online-Distributed-Channel} is $\mathcal{O}(NM)$, which increases linearly as the user size increases. \rev{Note that it would incur extra time overhead to obtain the social recommendation information for computing the Nash equilibrium. Nevertheless, in practice we do not need to update such information at every iteration slot since the length of a scheduling mini-slot for each iteration is at the time scale of microseconds \cite{jiang2010scheduling}, while the recommendation states (due to primary users' activities) change at the much larger time scale of milliseconds/seconds \cite{hossain2007cognitive}. In this case, for instance, we can update the social recommendation information only at the first iteration slot of the Nash equilibrium computation process, which would help to significantly reduce the time overhead. }

%\begin{algorithm}[H]
%\begin{itemize}
%\item \emph{Initialization}:
%\item \hspace{0.6cm}\emph{Set} initial channel selection profile $\boldsymbol{a}^{0}(t+1)=\boldsymbol{a}(t)$.
%\item \hspace{0.6cm}\emph{Set} the iteration index $l=0$.
%\item \emph{End Initialization}
%\end{itemize}
%$ $
%\begin{itemize}
%\item \emph{While} $\boldsymbol{a}^{l}(t+1)$ is not a Nash equilibrium
%\emph{do}:
%\item \hspace{0.6cm}\emph{For} secondary user $n=1$ to $N$ \emph{do}:
%\item \hspace{0.6cm}\hspace{0.6cm}\emph{Choose} the best channel $a_{n}^{l+1}(t+1)$
%that maximizes the expected throughput according to (\ref{eq:bs1}).
%\item \hspace{0.6cm}\hspace{0.6cm}\emph{Broadcast} the channel selection
%update $a_{n}^{l+1}(t+1)$ to other users.
%\item \hspace{0.6cm}\emph{End For}
%\item \hspace{0.6cm}\emph{Set} the iteration index $l=l+1$.
%\item \emph{End While}
%\end{itemize}
%\caption{\label{alg:Online-Distributed-Channel} Nash equilibrium computation
%algorithm for social recommendation aided DSA
%with strong network information}
%\end{algorithm}

\section{Social Recommendation Aided DSA: The Weak Information Case}\label{sec:Social-Recommendation-Based-1}

In this section we consider the social recommendation aided DSA mechanism design with the weak network information, by removing the assumptions in the strong information
case that a secondary user $n$ has the information of the channel
state transition matrix $(\Gamma_{m},\forall m\in\mathcal{M})$, the
mean channel throughput $(B_{m}^{n},\forall m\in\mathcal{M})$, and
the channel contention probabilities $(p_{k},\forall k\in\mathcal{N}_{n}^{p})$
of its interfering users. In the weak information case, a secondary
user $n$ only observes its channel recommendation states $\boldsymbol{I}_{n}(t)$
at each time slot $t$.

\subsection{DSA Using Distributed Reinforcement Learning}

For the weak network information case, we propose a distributed
reinforcement learning algorithm for social recommendation aided DSA, such that each secondary user learns to adjust its spectrum access strategy adaptively based on its local observations. A
key idea here is to extend the principle of single-agent reinforcement
learning \cite{bertsekas1995neuro} to a multi-agent setting. Such multi-agent reinforcement
learning algorithm has also been applied to the game theoretic models in \cite{cominetti2010payoff}. Here we generalize  the learning algorithm to the social recommendation aided DSA problem wherein a user's payoff depends on the interference graph structure and the environment is time varying (i.e., the channel recommendation states are varying in different time slots). This
leads to significant differences in analysis from those in \cite{bertsekas1995neuro,cominetti2010payoff}. For example,
we show that the convergence condition for the learning
algorithm depends on the structure of interference graph,
which is different from those results in \cite{bertsekas1995neuro,cominetti2010payoff}.

Specifically, a secondary user $n$ makes the spectrum access decision
based on the perception values $\boldsymbol{V}_{n}(t)=(V_{m,i}^{n},\forall m\in\mathcal{M},i\in\mathcal{I}\triangleq\{1,-1,0\})$,
where $V_{m,i}^{n}(t)$ represents the user $n$'s current perception
of the expected throughput of choosing channel $m$ when the channel
recommendation state $I_{m}^{n}(t)$ of channel $m$ is $i$. At the
channel selection stage of each time slot $t$, a
secondary user $n$ updates its perceptions based on the received
throughput $U_{n}(t)$ of the chosen channel $a_{n}(t)$ as
\begin{equation}
V_{m,i}^{n}(t)=\begin{cases}
(1-\alpha_{t})V_{m,i}^{n}(t-1)+\alpha_{t}U_{n}(t), & \mbox{if \ensuremath{a_{n}(t)=m}}\\
 & \mbox{and }I_{m}^{n}(t)=i,\\
V_{m,i}^{n}(t-1), & \mbox{otherwise,}
\end{cases}\label{eq:VV}
\end{equation}
where $\alpha_{t}$ are the smoothing factors satisfying that $\sum_{t}\alpha_{t}=\infty$ and $\sum_{t}\alpha_{t}^{2}<\infty$. Roughly speaking, (\ref{eq:VV}) implies  that a secondary user only changes the perception
of accessing the channel $m$ in recommendation state $I_{m}^{n}(t)$,
and keeps the perceptions of the channel $m$ in other recommendation
states $i\neq I_{m}^{n}(t)$ and other channels $m'\neq m$ in all recommendation
states unchanged.

Then, a secondary user $n$ makes the spectrum access decision and
chooses a channel $a_{n}(t+1)\in\mathcal{M}$ for the next time slot
$t+1$ according to the spectrum access strategy $\boldsymbol{\sigma}_{n}(t)=(\sigma_{1}^{n}(t),...,\sigma_{M}^{n}(t))$,
where $\sigma_{m}^{n}(t)$ is the probability of choosing channel
$m$. The strategy $\boldsymbol{\sigma}_{n}(t)$ is generated according
to the secondary user $n$'s current channel recommendation states $\boldsymbol{I}_{n}(t)$ and the
perception values $\boldsymbol{V}_{n}(t)$. Similar to the single-agent
learning, we choose the Boltzmann distribution \cite{bertsekas1995neuro} as the mapping from
perceptions to the spectrum access strategies, i.e.,
\begin{equation}
\sigma_{m}^{n}(t)=\frac{\exp\left(\beta V_{m,I_{m}^{n}(t)}^{n}(t)\right)}{\sum_{m'\in\mathcal{M}}\exp\left(\beta V_{m',I_{m'}^{n}(t)}^{n}(t)\right)},\forall m\in\mathcal{M},\label{eq:sss}
\end{equation}
where the parameter $\beta$ controls the degree of channel sampling (i.e., channel exploration). When $\beta$ is very small, (e.g., $\beta\rightarrow0$), each secondary  user tends to choose to access channels purely randomly (in this case we have maximum degree of channel exploration). When $\beta$ is large (e.g., $\beta\rightarrow\infty$), user $n$ will tend to exploit the channel with the current best perception value among all channels (in this case we have minimum degree of channel exploration). As shown in \cite{tekin2011online, tekin2012online},  the channel sampling plays a very important on the restless bandit based spectrum access problem. In general, without appropriate sampling the restless bandit algorithm may converge to sub-optimal solution, due to the sampling bias. Similarly, for our proposed learning mechanism, we will show later that a moderately small $\beta$ (which increases the randomness of the spectrum access strategy) is required to ensure sufficient channel sampling over the states so that user's perception values can be constantly updated and guarantee that the distributed learning converges to an efficient equilibrium with small performance loss.

We summarize the
distributed reinforcement learning algorithm for social recommendation
aided DSA in Algorithm \ref{alg:Distributed-reinforcement-learni}. Note that in Algorithm \ref{alg:Distributed-reinforcement-learni}, each user performs perception update and channel selection among $M$ channels in parallel. Hence Algorithm \ref{alg:Distributed-reinforcement-learni} has a very low computation complexity of $\mathcal{O}(M)$ in a time slot.
%
%\begin{algorithm}[H]
%\begin{itemize}
%\item \emph{Initialization}:
%\item \hspace{0.6cm}\emph{Set} set the initial perception values $V_{m,i}^{n}(0)=1$.
%\item \emph{End Initialization}
%\end{itemize}
%$ $
%\begin{itemize}
%\item \emph{Loop for the channel selection stage of each time slot $t$
%and each secondary user $n\in\mathcal{N}$ in parallel:}
%\item \hspace{0.6cm}\emph{Update }the perception values $V_{m,i}^{n}(t)$
%according to (\ref{eq:VV}).
%\item \hspace{0.6cm}\emph{Select} a channel $a_{n}(t+1)$ for the next
%time slot $t+1$ according to (\ref{eq:sss}).
%\item \emph{End Loop}
%\end{itemize}
%\caption{\label{alg:Distributed-reinforcement-learni}Distributed reinforcement
%learning algorithm for social recommendation aided spectrum access
%with weak network information}
%\end{algorithm}

\begin{algorithm}[tt]
\begin{algorithmic}[1]
\State \textbf{initialization:}
\State \hspace{0.4cm} \textbf{set} the initial perception values $V_{m,i}^{n}(0)=1$.
\State \textbf{end initialization \newline}
\Loop{ for the channel selection stage of each time slot $t$ and each secondary user $n\in\mathcal{N}$ in parallel:}
\State \textbf{update} the perception values $V_{m,i}^{n}(t)$ according to (\ref{eq:VV}).
\State \textbf{select} a channel $a_{n}(t+1)$ for the next time slot $t+1$ according to (\ref{eq:sss}).
\EndLoop
\end{algorithmic}
\caption{\label{alg:Distributed-reinforcement-learni}Distributed reinforcement
learning algorithm for social recommendation aided DSA
with weak network information}
\end{algorithm}

\subsection{Convergence of DSA Using Distributed Reinforcement Learning}

We now study the convergence of the proposed distributed reinforcement
learning algorithm. First, the perception value update in (\ref{eq:VV}) can
be written in the following equivalent form ($\forall n\in\mathcal{N},m\in\mathcal{M},i\in\mathcal{I}$),
\begin{equation}
V_{m,i}^{n}(t)-V_{m,i}^{n}(t-1)=\alpha_{t}\left(Z_{m,i}^{n}(t)-V_{m,i}^{n}(t-1)\right),\label{eq:zzz}
\end{equation}
where $Z_{m,i}^{n}(t)$ is the update value defined as
\[
Z_{m,i}^{n}(t)=\begin{cases}
U_{n}(t), & \mbox{if \ensuremath{a_{n}(t)=m}}\mbox{ and }I_{m}^{n}(t)=i,\\
V_{m,i}^{n}(t-1), & \mbox{otherwise.}
\end{cases}
\]

For the sake of brevity, we denote the perception values, update values,
spectrum access strategies, and channel recommendation states of all
the secondary users as $\boldsymbol{V}(t)\triangleq(V_{m,i}^{n}(t),\forall n\in\mathcal{N},m\in\mathcal{M},i\in\mathcal{I})$,
$\boldsymbol{Z}(t)\triangleq(Z_{m,i}^{n}(t),\forall n\in\mathcal{N},m\in\mathcal{M},i\in\mathcal{I})$,
$\boldsymbol{\sigma}(t)\triangleq(\sigma_{m}^{n}(t),\forall n\in\mathcal{N},m\in\mathcal{M}),$
and $\boldsymbol{I}(t)\triangleq(I_{m}^{n}(t),\forall n\in\mathcal{N},m\in\mathcal{M}),$
respectively.  Moreover, we define the mapping from the perceptions $\boldsymbol{V}(t)$ to
the conditional expected throughput of secondary user $n$ choosing channel
$m$, given the recommendation state $I_{m}^{n}(t)=i$, as \[
R_{m,i}^{n}(\boldsymbol{V}(t))\triangleq E\Bigl[U_{n}(t)|\boldsymbol{V}(t),I_{m}^{n}(t)=i,a_{n}(t)=m\Bigr],\]
where $E[\cdot]$ is taken with respect to the spectrum
access strategies $\boldsymbol{\sigma}(t)$ of all users (i.e., the
perceptions $\boldsymbol{V}(t)$ of all users due to (\ref{eq:sss})).
We have the following result.
\begin{lem}
\label{lem:For-the-distributed}For the distributed reinforcement
learning algorithm, if the parameter $\beta$ satisfies  the condition
\begin{equation}
\beta<\frac{1}{2\max_{n\in\mathcal{N},m\in\mathcal{M}}\{B_{m}^{n}\}\max_{n\in\mathcal{N}}|\mathcal{N}_{n}^{p}|},\label{eq:llll}
\end{equation}
the mapping from the perceptions to the expected throughput $\boldsymbol{R}(\boldsymbol{V}(t))\triangleq(R_{m,i}^{n}(\boldsymbol{V}(t)),\forall n\in\mathcal{N},m\in\mathcal{M},i\in\mathcal{I})$
forms a maximum-norm contraction mapping, i.e.,
\[
||\boldsymbol{R}(\boldsymbol{V})-\boldsymbol{R}(\hat{\boldsymbol{V}})||_{\infty}\leq\varepsilon||\boldsymbol{V}-\boldsymbol{\hat{V}}||_{\infty},
\]
 where $0<\varepsilon\triangleq2\beta\max_{n\in\mathcal{N},m\in\mathcal{M}}\{B_{m}^{n}\}\max_{n\in\mathcal{N}}|\mathcal{N}_{n}^{p}|<1$.
\end{lem}

The proof is relegated to the appendix. Lemma \ref{lem:For-the-distributed} implies that when the interference
among secondary users becomes more severe (i.e., the maximum degree
$\max_{n\in\mathcal{N}}|\mathcal{N}_{n}^{p}|$ of the interference
graph becomes larger), a smaller $\beta$ is needed to guarantee the
convergence. This is because that interference relationship among
users becomes more complicated and users should put more weight to
explore the environment. Based on the property of contraction mapping \cite{granas2003fixed}, we know that the sequence $\{\boldsymbol{V}(t),\forall t\geq0\}$ will converge to a fixed point $\boldsymbol{V}^{*}$.
\begin{thm}
\label{thm:For-the-distributed}For the distributed reinforcement
learning algorithm, if the parameter $\beta$ satisfies (\ref{eq:llll}), then the sequence $\{\boldsymbol{V}(t),\forall t\geq0\}$
converges to the solution $\boldsymbol{V}^{*}$ of the fixed point equation, which satisfies
that
\begin{equation}
R_{m,i}^{n}(\boldsymbol{V}^{*})=V_{m,i}^{n*},\forall n\in\mathcal{N},m\in\mathcal{M},i\in\mathcal{I}.\label{eq:tttttt}
\end{equation}
\end{thm}

The proof is given in the appendix. We next explore the property of the equilibrium $\boldsymbol{V}^{*}$
of the distributed reinforcement learning algorithm. Since $R_{m,i}^{n}(\boldsymbol{V}^{*})=E[U_{n}(t)|\boldsymbol{V}^{*},I_{m}^{n}(t)=i,a_{n}(t)=m]$,
from Theorem \ref{thm:For-the-distributed}, we see that a secondary
user $n$ can achieve an accurate estimation of the expected throughput
based on the perception value $V_{m,i}^{n*}$ at the equilibrium.
Moreover, for the spectrum access strategy $\boldsymbol{\sigma}^{*}$
generated according to (\ref{eq:sss}) based on perception values
$\boldsymbol{V}^{*}$, we can show the following result.
\begin{thm}
\label{thm:For-the-distributed2}For the distributed reinforcement learning algorithm, given a secondary user $n$'s channel
recommendation states $\boldsymbol{I}_{n}=(I_{m}^{n},\forall m\in\mathcal{M})$,
the spectrum access strategy $\boldsymbol{\sigma}_{n}^{*}$ at the equilibrium $\boldsymbol{V}^{*}$ maximizes its expected throughput approximately, i.e.,
\[
\sum_{m=1}^{M}\sigma_{m}^{n*}R_{m,I_{m}^{n}}^{n}(\boldsymbol{V}^{*})\geq\max_{\boldsymbol{\sigma}_{n}}\left\{ \sum_{m=1}^{M}\sigma_{m}^{n}R_{m,I_{m}^{n}}^{n}(\boldsymbol{V}^{*})\right\} -\psi,
\]
where the approximation gap $\psi$ is at most $\frac{1}{\beta}\ln M$.
\end{thm}
\begin{proof}
We first consider the following optimization problem:
\begin{eqnarray}
 & \max_{\boldsymbol{\sigma}_{n}} & \sum_{m=1}^{M}\sigma_{m}^{n}R_{m,I_{m}^{n}}^{n}(\boldsymbol{V}^{*})-\frac{1}{\beta}\sum_{m=1}^{M}\sigma_{m}^{n}\ln\sigma_{m}^{n},\label{eq:pppp}\\
 & \mbox{subject to} & \sum_{m=1}^{M}\sigma_{m}^{n}=1,\sigma_{m}^{n}\geq0,\forall m\in\mathcal{M}.\nonumber
\end{eqnarray}
The optimization formulation in (\ref{eq:pppp}) is motivated by that the distributed learning mechanism can converge to an equilibrium that drives a good balance between the channel exploration and channel exploitation. Thus, in the optimization formulation in (\ref{eq:pppp}) the term $\sum_{m=1}^{M}\sigma_{m}^{n}R_{m,I_{m}^{n}}^{n}(\boldsymbol{V}^{*})$ measures the performance of channel exploitation using the spectrum access strategy $\boldsymbol{\sigma}_{n}$ and the entropy term  $-\sum_{m=1}^{M}\sigma_{m}^{n}\ln\sigma_{m}^{n}$ measures the randomness of the spectrum access strategy $\boldsymbol{\sigma}$, which indicates the degree of channel exploration. Thus, the physical meaning of the optimization formulation in (\ref{eq:pppp}) is that we would like to find the optimal spectrum access strategy that strikes the best trade-off between the channel exploration and channel exploitation.  Since the problem is a convex optimization problem, by the KKT condition,
we can derive the optimal solution as
\[
\widetilde{\sigma}_{m}^{n}=\frac{\exp\left(\beta R_{m,I_{m}^{n}}^{n}(\boldsymbol{V}^{*})\right)}{\sum_{m'\in\mathcal{M}}\exp\left(\beta R_{m',I_{m'}^{n}}^{n}(\boldsymbol{V}^{*})\right)},\forall m\in\mathcal{M}.
\]
Since $R_{m,i}^{n}(\boldsymbol{V}^{*})=V_{m,i}^{n*}$ at the equilibrium
$\boldsymbol{V}^{*}$, we hence have that $\widetilde{\sigma}_{m}^{n}=\sigma_{m}^{n*},$
i.e., the spectrum access strategy $(\sigma_{m}^{n*},\forall m\in\mathcal{M})$
at the equilibrium $\boldsymbol{V}^{*}$ is the optimal solution to
the problem in (\ref{eq:pppp}). It follows that
\begin{eqnarray}
 &  & \sum_{m=1}^{M}\sigma_{m}^{n*}R_{m,I_{m}^{n}}^{n}(\boldsymbol{V}^{*})\nonumber \\
 & = & \max_{\boldsymbol{\sigma}_{n}}\left(\sum_{m=1}^{M}\sigma_{m}^{n}R_{m,I_{m}^{n}}^{n}(\boldsymbol{V}^{*})-\frac{1}{\beta}\sum_{m=1}^{M}\sigma_{m}^{n}\ln\sigma_{m}^{n}\right)\nonumber \\
 & & +\frac{1}{\beta}\sum_{m=1}^{M}\sigma_{m}^{n*}\ln\sigma_{m}^{n*}.\label{eq:pr5}
\end{eqnarray}
Furthermore, it is easy to check that
\begin{eqnarray}
 &  & \max_{\boldsymbol{\sigma}_{n}}\left(\sum_{m=1}^{M}\sigma_{m}^{n}R_{m,I_{m}^{n}}^{n}(\boldsymbol{V}^{*})-\frac{1}{\beta}\sum_{m=1}^{M}\sigma_{m}^{n}\ln\sigma_{m}^{n}\right)\nonumber \\
 & \geq & \max_{\boldsymbol{\sigma}_{n}}\sum_{m=1}^{M}\sigma_{m}^{n}R_{m,I_{m}^{n}}^{n}(\boldsymbol{V}^{*}).\label{eq:pr6}
\end{eqnarray}
Then from (\ref{eq:pr5}), (\ref{eq:pr6}), and due to the fact that  uniform distribution yields the largest entropy, i.e.,
\[-\sum_{m=1}^{M}\sigma_{m}^{n*}\ln\sigma_{m}^{n*}\leq\ln M,\] we know that
\begin{eqnarray*}
 &  & \sum_{m=1}^{M}\sigma_{m}^{n*}R_{m,I_{m}^{n}}^{n}(\boldsymbol{V}^{*})\\
 & \geq & \max_{\boldsymbol{\sigma}_{n}}\sum_{m=1}^{M}\sigma_{m}^{n}R_{m,I_{m}^{n}}^{n}(\boldsymbol{V}^{*})+\frac{1}{\beta}\sum_{m=1}^{M}\sigma_{m}^{n*}\ln\sigma_{m}^{n*}\\
 & \geq & \max_{\boldsymbol{\sigma}_{n}}\sum_{m=1}^{M}\sigma_{m}^{n}R_{m,I_{m}^{n}}^{n}(\boldsymbol{V}^{*})-\frac{1}{\beta}\ln M.
\end{eqnarray*}
\end{proof}
%The proof is given in Appendix \ref{proof3}.

Theorems \ref{thm:For-the-distributed} and \ref{thm:For-the-distributed2}
together illustrate the trade-off between the channel exploration and exploitation
through the choice of $\beta$. As we have shown in Theorem \ref{thm:For-the-distributed}, a moderately small $\beta$ is required
to ensure the necessary degree of channel sampling for network environment exploration so that user's perception values can be constantly updated, and guarantees the convergence of distributed learning mechanism to an equilibrium. Furthermore, by deriving the approximation gap in Theorem \ref{thm:For-the-distributed2}, we show that when a moderate $\beta$ is used, the equilibrium by the distributed learning mechanism is an efficient approximate solution with a small performance loss.

%As per Theorems \ref{thm:For-the-distributed}, a small enough $\beta$ is required
%to explore the network environment (so that users will not put too many weights
%on exploitation and get stuck in channels with the current best perception)
%and guarantee the convergence of distributed learning
%to the equilibrium. However, as illustrated in  Theorems \ref{thm:For-the-distributed2}, if $\beta$ is too small, then the performance loss of the distributed learning would grow large due to over-exploration. Thus, a moderate parameter $\beta$ can be adopted to drive the balance between the channel exploration and exploitation, such that the learning mechanism can converge to an equilibrium, which is an efficient approximate solution with a small performance loss.}

%A small enough $\beta$ is required
%to explore the environment (so that users will not put too many weights
%on exploitation and get stuck in channels with the current best throughputs)
%and guarantee the convergence of distributed reinforcement learning
%to the equilibrium. If $\beta$ is too small, however, then the approximation
%gap $\psi$ grows large due to over-exploration.

\section{Further Discussions}\label{Discusssion}
We have studied social recommendation aided DSA with both strong and weak network information in Sections \ref{sec:Social-Recommendation-Based} and \ref{sec:Social-Recommendation-Based-1}, respectively. For the strong network information case, to meet the requirement that a secondary user has strong knowledge
of the network environment, a supporting infrastructure (e.g., secondary access point) to collect and analyze the statistical network information is typically needed, which would demand significant implementation overhead and high deployment cost in practice. This motivates us to explore the weak network information case and devise a more practical social recommendation aided DSA mechanism, such that each secondary user adapts its spectrum access decision based on local observations. Another major motivation of considering the strong network information case is that it can serve as the benchmark for the weak information case. Performance evaluation in Section \ref{Sim} demonstrates  that the social recommendation aided DSA mechanism with weak network information can achieve superior performance, with a small performance loss (at most $12\%$) with respect to the strong network information case.

An innovative feature of the proposed mechanisms in this paper is that we explore the social dimension of DSA and leverage the endogenous social incentive which comes from the intrinsic social ties among secondary users. The underlying
rationale is that hand-held devices are carried by human beings and the social trust stemming from human social ties (e.g., kinship, friendship, or colleague relationship) can be utilized to stimulate effective and trustworthy collaboration for channel recommendation. Building upon the social trust among secondary users, we can prevent the potential attacks of releasing false channel recommendations by untrusted/unacquainted users and enhance the security level of DSA. Moreover, by participating in collaborative channel recommendation with social neighbors,  each individual user can receive exogenous incentive by benefiting from possessing a better view of the entire network environment and achieving more informed spectrum access decision making. Numerical results in Section \ref{Sim} corroborate that the performance of social recommendation aided DSA improves as the social link density among the secondary users increases (due to the increase of cooperative social interactions).

As discussed in Section \ref{socialgraph}, to preserve the privacy among secondary users, the social relationship identification
procedure can be carried out prior to the spectrum access. Specifically, two secondary users in proximity can locally initiate the ``matching" process to detect the common social features between them. For example, two users can match their contact lists. If they have the phone numbers of each other or many of their phone numbers are the same, then it is very likely that they know each other. As another example, two device users can match their home and working addresses and identify whether they are neighbors or colleagues. To preserve the privacy of the secondary users, the private matching protocol in \cite{von2008veneta} can be adopted to design a privacy-preserving social relationship identification mechanism.  Moreover, to thwart the tampering and falsification attacks on the channel recommendation messages by malicious users, we can adopt the techniques of attribute-based encryption and signatures in \cite{bethencourt2007ciphertext} and \cite{maji2011attribute} such that two secondary users with social trust (e.g., family members)  can encrypt/decrypt and authenticate the messages based on their common and shared social features (e.g., home address and birthday). The focus of this paper is to devise efficient social recommendation aided DSA mechanisms by considering different network information settings. Due to space limitation, we will pursuit a thorough understanding of the security and privacy issues of the social recommendation aided DSA in future work.

\begin{figure*}[tt]
\begin{minipage}[t]{0.32\linewidth}
\centering
\includegraphics[scale=0.38]{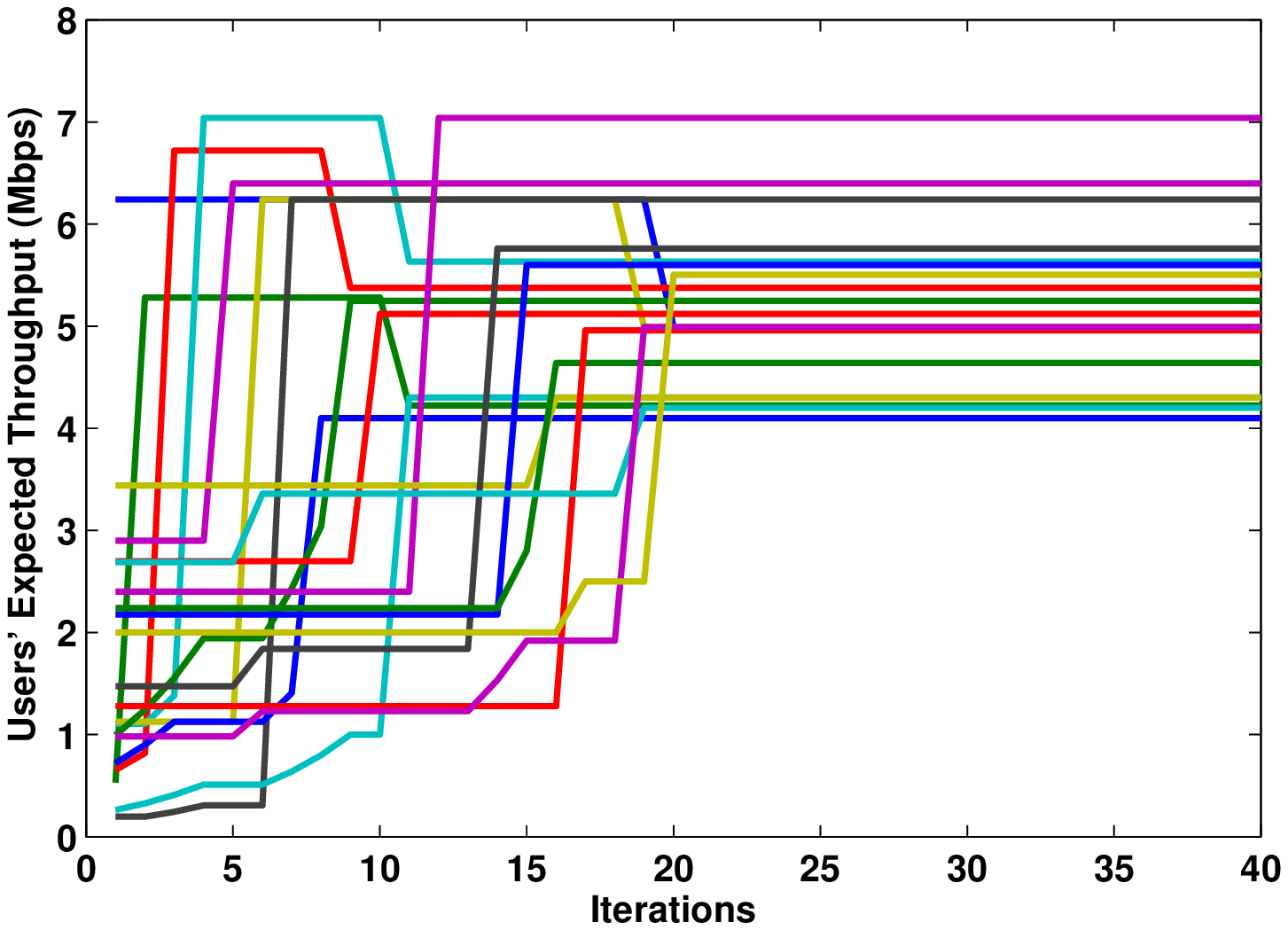}
\caption{\label{fig:Dynamic-of-users'}Dynamics of users' expected throughput
by Nash equilibrium computation algorithm}
\end{minipage}
\hfill
\begin{minipage}[t]{0.32\linewidth}
\centering
\includegraphics[scale=0.38]{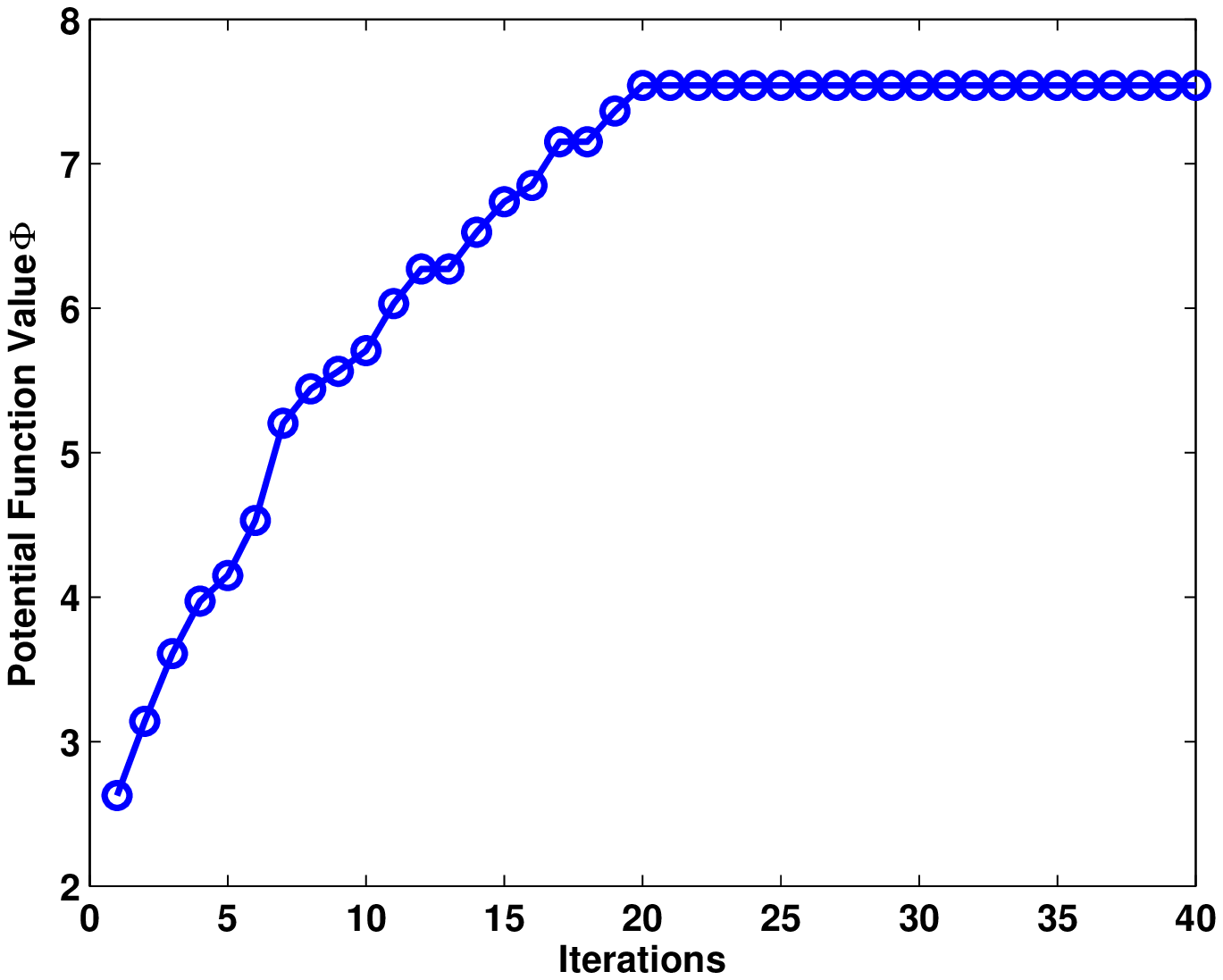}
\caption{\label{fig:Dynamics-of-potential}Dynamics of potential function value
by Nash equilibrium computation algorithm}
\end{minipage}
\hfill
\begin{minipage}[t]{0.32\linewidth}
\centering
\includegraphics[scale=0.38]{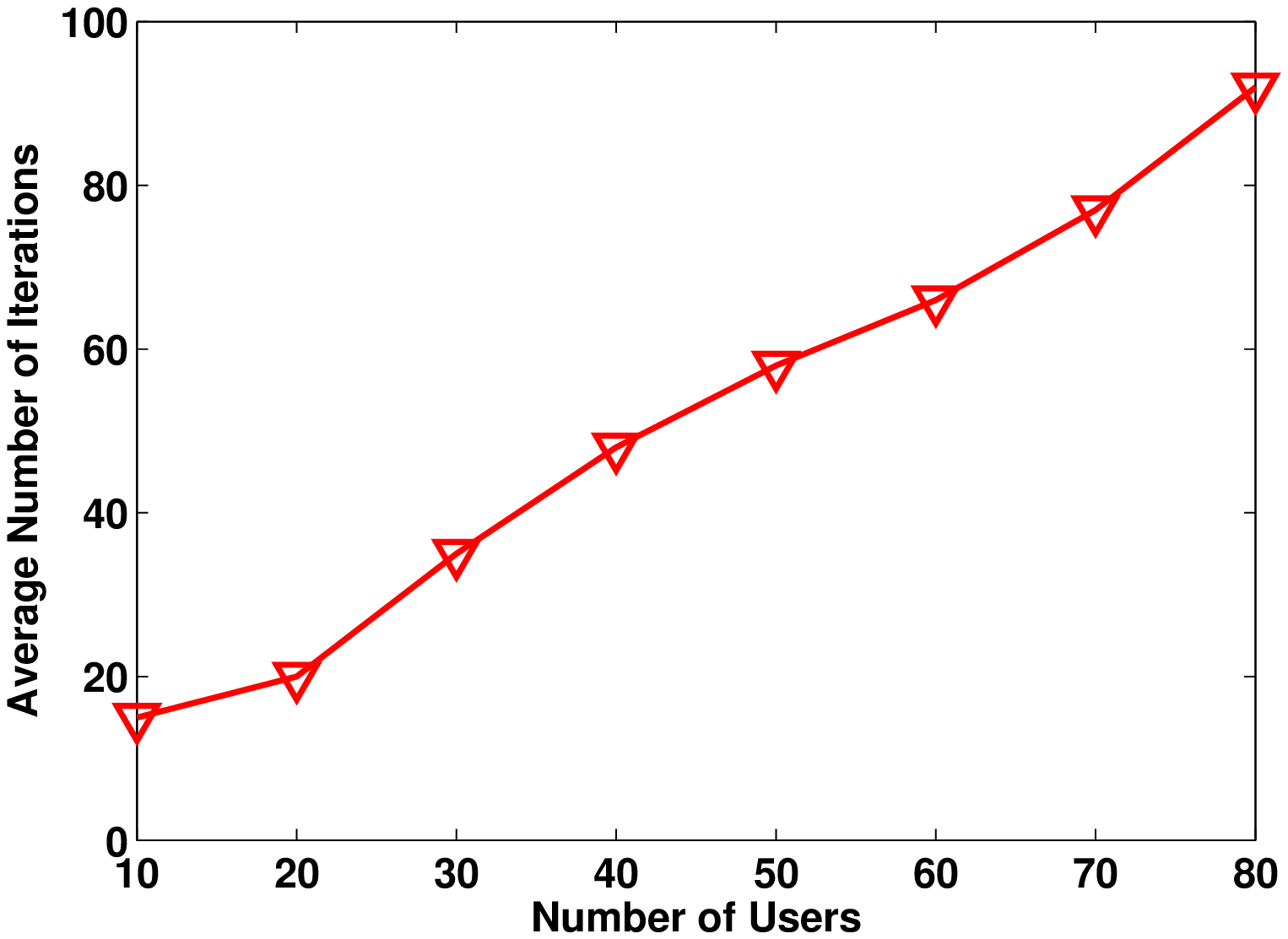}
\caption{\label{fig:Iterations}Average number of iterations for convergence by Nash equilibrium computation algorithm}
\end{minipage}
\end{figure*}

\begin{figure*}[tt]
\begin{minipage}[t]{0.32\linewidth}
\centering
\includegraphics[scale=0.38]{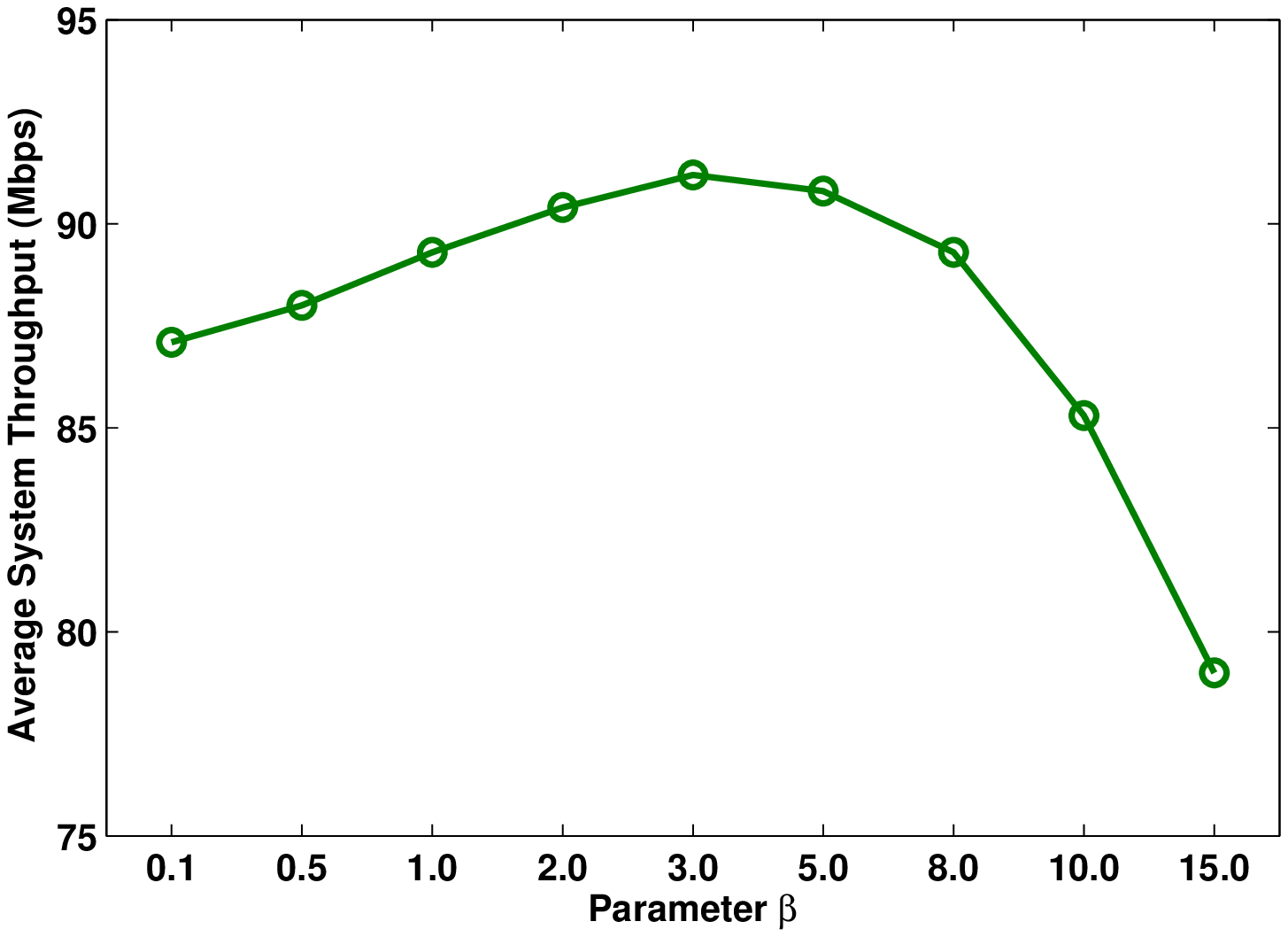}
\caption{\label{fig:System-performance-of}System performance of the distributed
reinforcement learning with different parameters $\beta$}
\end{minipage}
\hfill
\begin{minipage}[t]{0.32\linewidth}
\centering
\includegraphics[scale=0.38]{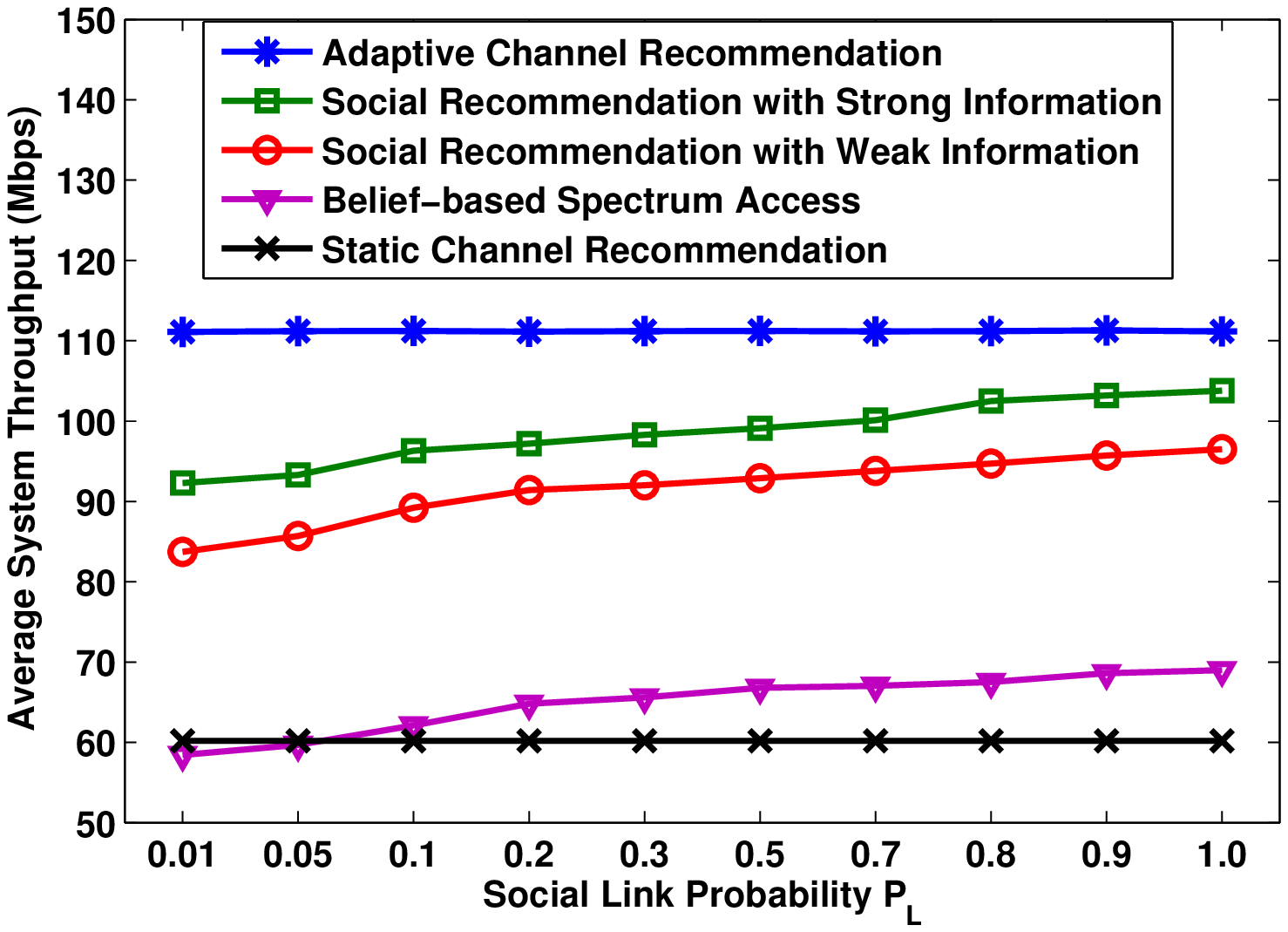}
\caption{\label{fig:System-throughput-with}System throughput with $N=20$
secondary users and different social link probabilities}
\end{minipage}
\hfill
\begin{minipage}[t]{0.32\linewidth}
\centering
\includegraphics[scale=0.38]{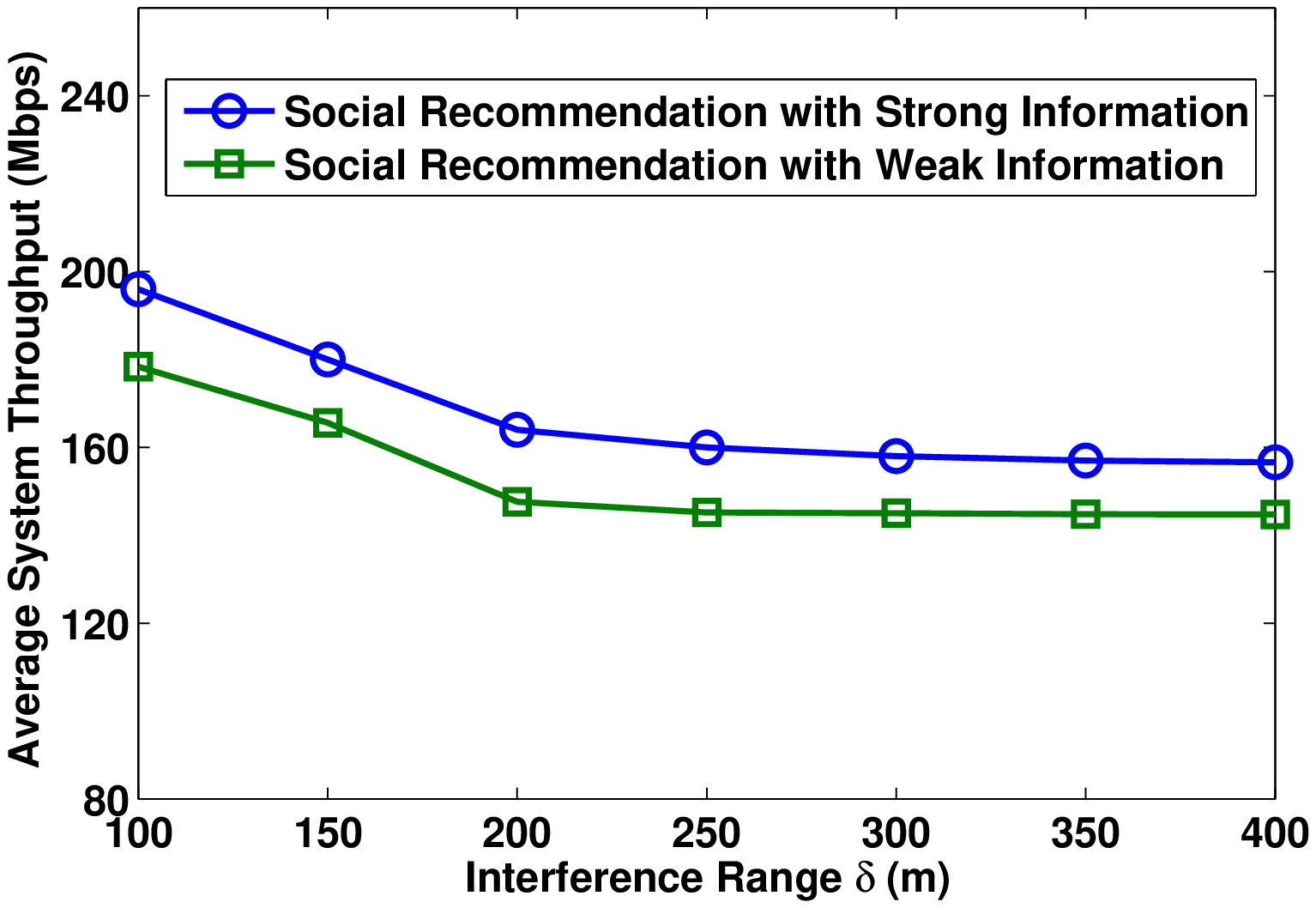}
\caption{\label{fig:Interference}System throughput with $N=60$ secondary users and different interference ranges}
\end{minipage}
\end{figure*}

\section{Numerical Results}\label{Sim}

In this section, we evaluate the performance of the proposed social
recommendation aided DSA schemes with strong and weak
network information by numerical studies.

We first consider a dynamic spectrum access network that consists of
$N=20$ secondary users and $M=5$ primary channels. For each channel
$m$, we set the channel state transition probabilities $\lambda_{m}=\mu_{m}=0.2$.
We consider a Rayleigh fading channel environment such that the channel
gain of a secondary user $n$ on an idle channel $m$ in each time
slot $t$ follows the exponential distribution with the mean channel
data rate $B_{m}^{n}$ being randomly assigned from the set $\{10,20,..,50\}$
Mbps. For the idle channel contention, the contention probability
$p_{n}$ of a secondary user will be randomly assigned from the set
$\{0.1,0.2,0.3\}$. We consider that the secondary users are randomly scattered across a square area with a side length of $500$
m and construct the interference graph $\mathcal{G}^{p}$ by setting
the interference range $\delta=100$ m, i.e., there is an interference
edge between two secondary users if their distance is less than $100$
m. For the social graph $\mathcal{G}^{s}$ for channel recommendation,
we will consider two types of social graphs: Erdos-Renyi social graph
and the real data trace based social graph.

\subsection{Erdos-Renyi Social Graph}

We first consider the case that the social graph $\mathcal{G}^{s}$
is represented by Erdos-Renyi (ER) graph model \cite{newman2002random} where a
social link exists between any two secondary users with a probability
of $P_{L}$.

We first show the dynamics of the Nash equilibrium computation algorithm
in Algorithm \ref{alg:Online-Distributed-Channel} for the social
recommendation aided DSA with strong network
information in Figure \ref{fig:Dynamic-of-users'}. We see that the
algorithm can converge to an equilibrium in less than $30$ steps
of iterations. To verify that the equilibrium is a Nash equilibrium,
we show the dynamics of the potential function value $\Phi(\boldsymbol{a})$
of the social recommendation aided DSA game in Figure
\ref{fig:Dynamics-of-potential}. We see that the algorithm can drive
the potential function to a maximal point, which is a Nash equilibrium
according to the property of potential game. We also show the average number of iterations for convergence by Algorithm \ref{alg:Online-Distributed-Channel} in a time slot in Figure \ref{fig:Iterations} with different number of users $N=10,20,...,80$. We see that the algorithm increases linearly as the number of users increases. This demonstrates the algorithm can scale well as the system size increases. This implies that Algorithm \ref{alg:Online-Distributed-Channel} has a very low complexity in practice. For example, in practical network systems each time slot has many scheduling mini-slots, and each of which has a length of several microseconds (e.g., 9 microseconds for 802.11a) \cite{jiang2010scheduling}. In this case, if in the social recommendation system we allow a time budget of 1 millisecond for computing the Nash equilibrium during a time slot, then a secondary base-station can accommodate over 100 associated users.

Note that we consider a slotted time structure and assume that the channel state is fixed during a time slot (but may change across different time slots). As long as such an assumption holds, during each time slot we can still compute the Nash equilibrium given users' received recommendation states using Algorithm 1. If the channel availability changes at a fast time scale, then a shorter time slot should be defined in the system to ensure the assumption above is valid. In this case, the allowable time budget for computing the Nash equilibrium during a time slot may be shortened and hence the number of associated users that a secondary base-station can accommodate may decrease. Nevertheless, the activities of primary users typically change over the time scale of milliseconds \cite{liang2011cognitive},  which implies that we have sufficient time budget for computing the Nash equilibrium during a time slot in practice. For the distributed learning algorithm in the weak information case, each user updates its strategy locally based on its recommendation information during a time slot, and hence the time overhead for such an update can be ignored. Thus, the distributed learning algorithm can also well perform.

We then evaluate the distributed reinforcement learning algorithm
with different choices of parameter $\beta$ for social recommendation
aided DSA with weak network information. The
result in Figure \ref{fig:System-performance-of} shows the system
performance with different parameters $\beta$, and demonstrates that
a proper parameter can achieve a balance between exploration and exploitation
and offer the best performance. When $\beta$ is small, the secondary
users tend to select the channels more randomly (i.e., over-exploration)
and the performance gap can be large according to Theorem \ref{thm:For-the-distributed2}.
When $\beta$ is very large, the algorithm may not converge
due to over-exploitation and the performance is again negatively affected.
In the following simulations we set $\beta=3$ since it can achieve
a good system performance.

%\begin{figure*}[tt]
%\begin{minipage}[t]{0.32\linewidth}
%\centering
%\includegraphics[scale=0.35]{User}
%\caption{\label{fig:Dynamic-of-users'}Dynamics of users' expected throughput
%by Nash equilibrium computation algorithm}
%\end{minipage}
%\hfill
%\begin{minipage}[t]{0.32\linewidth}
%\centering
%\includegraphics[scale=0.35]{Potential}
%\caption{\label{fig:Dynamics-of-potential}Dynamics of potential function value
%by Nash equilibrium computation algorithm}
%\end{minipage}
%\hfill
%\begin{minipage}[t]{0.32\linewidth}
%\centering
%\includegraphics[scale=0.35]{Learning}
%\caption{\label{fig:System-performance-of}System performance of the distributed
%reinforcement learning with different parameters $\beta$}
%\end{minipage}
%\end{figure*}
%
%
%\begin{figure*}[tt]
%\begin{minipage}[t]{0.32\linewidth}
%\centering
%\includegraphics[scale=0.35]{ER}
%\caption{\label{fig:System-throughput-with}System throughput with $N=20$
%secondary users and different social link probabilities}
%\end{minipage}
%\hfill
%\begin{minipage}[t]{0.32\linewidth}
%\centering
%\includegraphics[scale=0.35]{Link}
%\caption{\label{fig:Number-of-social}Number of social links of the social
%graphs based on real trace Brightkite}
%\end{minipage}
%\hfill
%\begin{minipage}[t]{0.32\linewidth}
%\centering
%\includegraphics[scale=0.35]{Brightkite}
%\caption{\label{fig:System-throughput-with-1}System throughput with different
%number of secondary users}
%\end{minipage}
%\end{figure*}

To evaluate the impact of social link density of the social graph,
we implement the simulations with different social link probabilities
$P_{L}=0.01,0.05,...,1.0$, respectively. For each given $P_{L}$,
we average over $100$ runs. As the benchmark, we also compare the
social recommendation aided DSA mechanisms with the following schemes:

(1) \emph{Static channel recommendation}: We implement the static channel
recommendation scheme in \cite{li2010customer}, such that all the secondary users
are assumed to be cooperative to recommend channels to each other
and a user will put more weight on selecting a recommended channel
according to a constant branching probability $P_{rec}$. Specifically,
suppose a user has $R$ recommended channels in a time slot, then
the user will choose a recommended channel with a probability of $\frac{P_{rec}}{R}$
and an un-recommended channel with a probability $\frac{P_{rec}}{M-R}$.
We obtain the optimal constant branching probability $P_{rec}$ by
exhaustive search.

(2) \emph{Adaptive channel recommendation}: We implement the adaptive channel
recommendation scheme in \cite{ACR2013}, such that all the secondary users
will select channels adaptively in order to maximize the average system
throughput. The optimal channel selection policies of all the secondary users
are computed by the global optimization approach -- Model Reference
Adaptive Search method \cite{hu2007model}. Note that the adaptive channel recommendation
scheme provides the performance upper-bound for the proposed social
recommendation mechanisms, since it maximizes the system-wide performance
by assuming that the network information is complete and all the secondary
users are cooperative (i.e., the interest of all users is aligned
and the underlying social graph for channel recommendation is complete).

(3) \emph{Belief-based spectrum access}: We implement the belief-based
spectrum access scheme in \cite{lai2011cognitive2}, such that each user shares the
sensing result of the chosen channel with its social neighbors at
the end of each time slot. Accordingly, a user $n$ updates two vectors
$\boldsymbol{X}_{n}=(X_{1}^{n},...,X_{M}^{n})$ and $\boldsymbol{Y}_{n}=(Y_{1}^{n},...,Y_{M}^{n})$
with $X_{m}^{n}$ and $Y_{m}^{n}$ denoting the number of time slots
in which channel $m$ is idle, and the number of time slots in which
the channel $m$ has been accessed, respectively. Then a user $n$
computes its belief as $\nu_{m}^{n}=\frac{X_{m}^{n}}{Y_{m}^{n}}$
and chooses a channel $m$ with a probability $\frac{\nu_{m}^{n}}{\sum_{i=1}^{M}\nu_{i}^{n}}.$

We show the average system throughput by all the schemes in Figure
\ref{fig:System-throughput-with}. We see that the performance of
social recommendation schemes with strong and weak network information
increases as the social link probability $P_{L}$ increases. This is due to the fact that as the increase of cooperative social interactions, each individual secondary user benefits more from the channel recommendation by having a better view of the network environment and achieving more informed spectrum access decision making. Compared
with the adaptive channel recommendation scheme, the performance losses
of social recommendation schemes with strong and weak network information
are at most $16\%$ and $25\%$, respectively. This demonstrates that the proposed schemes can still achieve a good performance even when the number of social friends for channel recommendations is small. When the social link
probability $P_{L}$ is large, the performance losses can be further
reduced to $7\%$ and $14\%$, respectively. We also observe that
the social recommendation schemes with strong and weak network information
can achieve at least $48\%$ and $39\%$ performance gain over the
belief-based spectrum access scheme in all cases. Compared with the
static channel recommendation scheme, the performance gains of social
recommendation schemes with strong and weak network information are
up-to $73\%$ and $60\%$, respectively. Furthermore, compared with the case of strong network information, the performance loss of social recommendation aided DSA with weak network information is at most $12\%$

We next investigate the impact of interference range $\delta$ on the system performance. We fix the number of users $N=60$ and the social link probability $P_{L}=0.2$ and set different interference ranges $\delta=100, 150, ..., 400$m, respectively. We show the average system performance by both social recommendation aided DSA mechanisms in Figure \ref{fig:Interference}. We observe that the system performance decreases as the interference range $\delta$ increases. This is because the increase of  the interference range leads to severe interferences among users for spectrum access. When interference range is large enough (e.g., $\delta \geq 300$m), the physical graph is very close to the fully meshed case, and hence the system degradation levels off.

%\begin{figure*}[tt]
%\begin{minipage}[t]{0.48\linewidth}
%\centering
%\includegraphics[scale=0.4]{Link}
%\caption{\label{fig:Number-of-social}Number of social links of the social
%graphs based on real trace Brightkite}
%\end{minipage}
%\hfill
%\begin{minipage}[t]{0.48\linewidth}
%\centering
%\includegraphics[scale=0.4]{Brightkite}
%\caption{\label{fig:System-throughput-with-1}System throughput with different
%number of secondary users}
%\end{minipage}
%\end{figure*}

\begin{figure}[tt]
\centering
\includegraphics[scale=0.4]{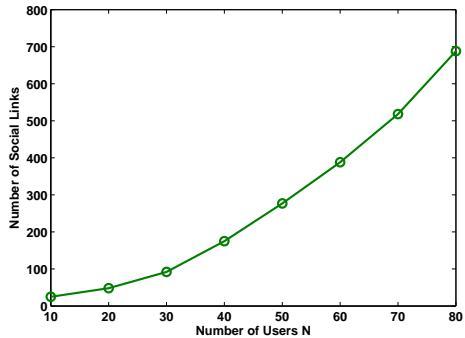}
\caption{\label{fig:Number-of-social}Number of social links of the social
graphs based on real trace Brightkite}
\end{figure}

\begin{figure}[tt]
\centering
\includegraphics[scale=0.4]{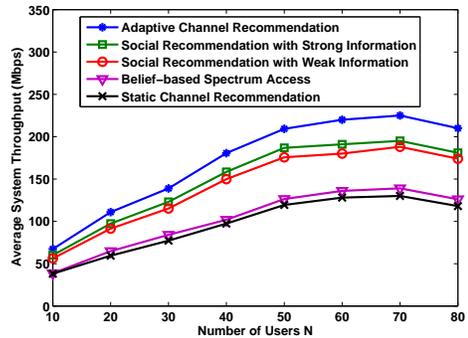}
\caption{\label{fig:System-throughput-with-1}System throughput with different
number of secondary users}
\end{figure}

\subsection{Real Trace Based Social Graph}

We then evaluate the proposed social recommendation schemes with the social graphs for channel recommendation
generated according to the friendship network of the real data trace
Brightkite \cite{KONECT2013}. Brightkite contains
an explicit friendship network among the users. Different from
the Erdos-Renyi (ER) social graph, the friendship network of
Brightkite is scale-free such that the node degree distribution
follows a power law \cite{KONECT2013}. We implement simulations with the number of secondary
users $N=10,20,...,80$, respectively. The total number of social
links among these users of the social graphs is shown in Figure \ref{fig:Number-of-social}. For a given social graph with a fixed user size, we randomly assign the node IDs in the social graph to the secondary users in the social recommendation DSA system and repeat over $100$ times to compute the average performance.

We show the average system throughput in Figure \ref{fig:System-throughput-with-1}.
We see that the performance losses of social recommendation schemes
with strong and weak network information are at most $12\%$ and $17\%$,
respectively, compared with the adaptive channel recommendation scheme.
Moreover, the social recommendation schemes with strong and weak network
information can achieve up-to $55\%$ and $45\%$ performance gain
over the belief-based spectrum access scheme, respectively. Compared
with the static channel recommendation scheme, the performance gains
of social recommendation schemes with strong and weak network information
are up-to $63\%$ and $54\%$, respectively. The performance loss of social recommendation aided DSA with weak network information is at most $8\%$ with respect to the strong network information case. This demonstrates the
efficiency of the proposed social recommendation aided DSA mechanisms when the real social network is used in practices.

\rev{We also observe from Figure \ref{fig:System-throughput-with-1} that when the network size is small (e.g., $N<50$), the system-wide performance grows almost linearly. This is mainly due to that the performance gain from the spectrum spatial reuse (i.e., more spectrum opportunities are utilized at more locations) is more significant than the increase of the cost by information exchange and interference. When the the network size becomes larger (e.g., $N>50$), the system-wide performance grows sub-linearly and then max out beyond a certain network size (e.g., $N=70$). This is because that the cost by information exchange and severe interference among users becomes very heavy and dominates the gain by the spectrum spatial reuse when the network size is large.}

%\section{Related Work}\label{relatedwork}
%The social aspect is now becoming an important dimension
%for communication system design. Social structures,
%such as social community which are derived from the user
%contact patterns, have been exploited to design efficient data
%forwarding and dissemination algorithms in delay tolerant networks
%\cite{costa2008socially,gao2009multicasting}. Two key social social phenomena of social trust and social reciprocity are exploited for cooperative D2D communication mechanism design in \cite{chen2013social}. In this paper, we explore social intelligence of secondary users for distributed spectrum access via channel recommendation.

\section{Conclusion}\label{Con}
In this paper, we study the social recommendation aided DSA mechanism design in both strong and weak network information cases. For the strong network information case, we model the social recommendation aided DSA problem at each time slot as a strategic game and exploit the structure property of the game to devise the Nash equilibrium computation algorithm. For the weak network information case, we develop a distributed reinforcement learning mechanism using only the local observations of secondary users. We also derive the convergence conditions and characterize the equilibrium of the learning mechanism. The efficacy of the proposed mechanisms is further demonstrated through extensive numerical evaluations using real social data trace.

For future work, we will consider related security and privacy issues. We are utilizing the techniques such as private matching and attribute-based encryption and signatures to design secure and privacy-preserving social recommendation mechanisms for DSA.

\bibliographystyle{ieeetran}
\bibliography{DynamicSpectrum,SocialSpectrum}

\end{document}